\documentclass[%
 reprint,
nofootinbib,
 amsmath,amssymb,
 aps,
 prx,
]{revtex4-2}

\usepackage{graphicx}
\usepackage{dcolumn}
\usepackage[english]{babel}

\usepackage{bm}
\usepackage{amsthm}
\usepackage[dvipsnames]{xcolor}

\newcommand{\inv}{^{-1}}
\newcommand {\R}{\mathbb R}

\newcommand{\abk}[1]{\left\langle #1 \right\rangle}

\renewcommand{\vec}[1]{\mathbf{#1}}
\newcommand{\unit}[1]{\ \mathrm{ #1}}
\DeclareMathOperator{\supp}{supp}
\DeclareMathOperator{\FP}{FP}
\DeclareMathOperator{\exc}{exc.}
\DeclareMathOperator{\nullh}{null}

\theoremstyle{definition}

\newtheorem{prop}{Proposition}

\begin{document}


\title{Diverse mean-field dynamics of clustered, inhibition-stabilized Hawkes networks via combinatorial threshold-linear networks}

\author{Caitlin Lienkaemper}
 \altaffiliation{clienk@mit.edu}
 \affiliation{
Massachusetts Institute of Technology, Department of Brain and Cognitive Science }

\author{Gabriel Koch Ocker}%
 \email{gkocker@bu.edu}
\affiliation{Boston University, Department of Mathematics and Statistics and Center for Systems Neuroscience}

\date{\today}

\begin{abstract}
  Networks of interconnected neurons display diverse patterns of collective activity. Relating this collective activity to the network's connectivity structure is a key goal of computational neuroscience. 
We approach this question for clustered  networks, which can form via biologically realistic learning rules and allow for the re-activation of learned patterns. Previous studies of clustered networks have focused on metastabilty between fixed points, leaving open the question of whether clustered spiking networks can display more rich dynamics--and if so, whether these can be predicted from their connectivity. 
Here, we show that in the limits of large population size and fast inhibition, the combinatorial threshold linear network (CTLN) model is a mean-field theory for inhibition-stabilized nonlinear Hawkes networks with clustered connectivity. 
The CTLN has a large body of ``graph rules'' relating network structure to dynamics. By applying these, we can predict the dynamic attractors of our clustered spiking networks from the structure of between-cluster connectivity. This allows us to construct networks displaying a diverse array of nonlinear cluster dynamics, including metastable periodic orbits and chaotic attractors. Relaxing the assumption that inhibition is fast, we see that the CTLN model is still able to predict the activity of clustered spiking networks with reasonable inhibitory timescales. For slow enough inhibition, we observe bifurcations between CTLN-like dynamics and global excitatory/inhibitory oscillations. 
\end{abstract}

\maketitle


\section{\label{sec:intro} Introduction}
The dynamics of neuronal networks implement sensory, motor, and cognitive computations. To understand how those dynamics are governed by the underlying network structure is a central goal of computational and systems neuroscience. We aim to unite two theoretical approaches to this goal: mean-field theories which relate the statistics of activity to the statistics of connectivity in large networks, and graph-theoretic approaches which relate the attractors of a network to its particular structure.

The first approach relates the connectivity statistics of large, randomly connected recurrent networks to macroscopic patterns or statistics of their dynamics~\cite{ocker_statistics_2017, la_camera_mean_2022, galves_probabilistic_2024}.  
Here, we consider networks with cluster-like structure: groups of densely and/or strongly interconnected neurons. Clustered structure is  found in cortical networks \cite{markram_network_1997, perin_synaptic_2011, yoshimura_excitatory_2005} and can form via biologically realistic learning rules like spike-time dependent plasticity \cite{litwin-kumar_formation_2014, ocker_training_2019, perin_synaptic_2011, zenke_diverse_2015, montangie_autonomous_2020}. 
Clustered networks are a popular model for metastability (long-lived but transient regimes of activity), especially metastability between high and low activity states in the neocortex \cite{brinkman_metastable_2022, litwin-kumar_slow_2012, rossi_dynamical_2024, mazzucato_expectation-induced_2019}. This metastability is hypothesized to be the dynamical mechanism for working memory, perceptual multistability, and decision-making~\cite{amit_model_1997, wang_decision_2008, moreno-bote_noise-induced_2007, albantakis_changes_2011}. In large networks, the macroscopic dynamics are sensitive to statistics of the network structure but not changes in individual synaptic weights.

In small networks, however, minor changes in connectivity can cause major changes in dynamics \cite{marder_central_2001, curto_graph_2023}. 
Recently, combinatorial threshold-linear network models (CTLNs) have proven a fruitful setting for understanding the link between graph structure and dynamics \cite{curto_pattern_2016, curto_fixed_2019}. 
CTLNs are a simple model but can exhibit diverse nonlinear dynamics \cite{morrison_diversity_2024}.
That simplicity exposes 1) rules for CTLNs that relate their fixed points and dynamic attractors to their graph structure and 2) a theory of compositionality that describes how attractors in larger networks relate to attractors of their subgraphs \cite{curto_fixed_2019, curto_stable_2024, curto_graph_2023, parmelee_sequential_2022, parmelee_core_2022, santander_nerve_2022}.
This approach has focused on particular graph structures with at most tens of nodes. 
Results from this setting can, however, advance our understanding of which features of connectivity play a major role in determining dynamics and should be incorporated into statistical approaches. 

In this paper, we investigate how the dynamics of a point process network model 
relates to the clustered structure of its excitatory connectivity. We show that networks with structured connectivity between clusters can exhibit varied dynamics: in addition to metastability between fixed points, these networks can exhibit periodic orbits and chaotic attractors, as well as metastability between these dynamic attractors. To relate these dynamics to the network's macroscopic structure, we show that the CTLN is dynamically equivalent to a mean-field theory for the cluster-mean firing rates in the limits of large population size and fast inhibition. This allows us to apply graph rules for CTLNs to predict the dynamics of our clustered spiking networks. 

Finally, we examine how the timescale of inhibitory activity interacts with the cluster structure to shape the network dynamics.
While the CTLN model is informally motivated as describing the activity of a population of excitatory neurons with specific connectivity against a pattern of diffuse inhibition, similar to the ``blanket of inhibition" described in \cite{fino_dense_2011}, a formal derivation of the CTLN model from networks with explicit excitation and inhibition has not been published.
We show that while the CTLN model formally describes the limit where inhibition is very fast, it still qualitatively describes the dynamics of excitatory-inhibitory networks with reasonable inhibitory timescales.  For sufficiently slow inhibition, we observe bifurcations between CTLN-like dynamics and global excitatory/inhibitory oscillations similar to the pyramidal/interneuron gamma (PING) oscillation \cite{whittington_inhibition-based_2000}. 

\section{Results}
\subsection{Clustered spiking networks display diverse dynamics.}

\begin{figure*}[ht!]
    \centering
    \includegraphics{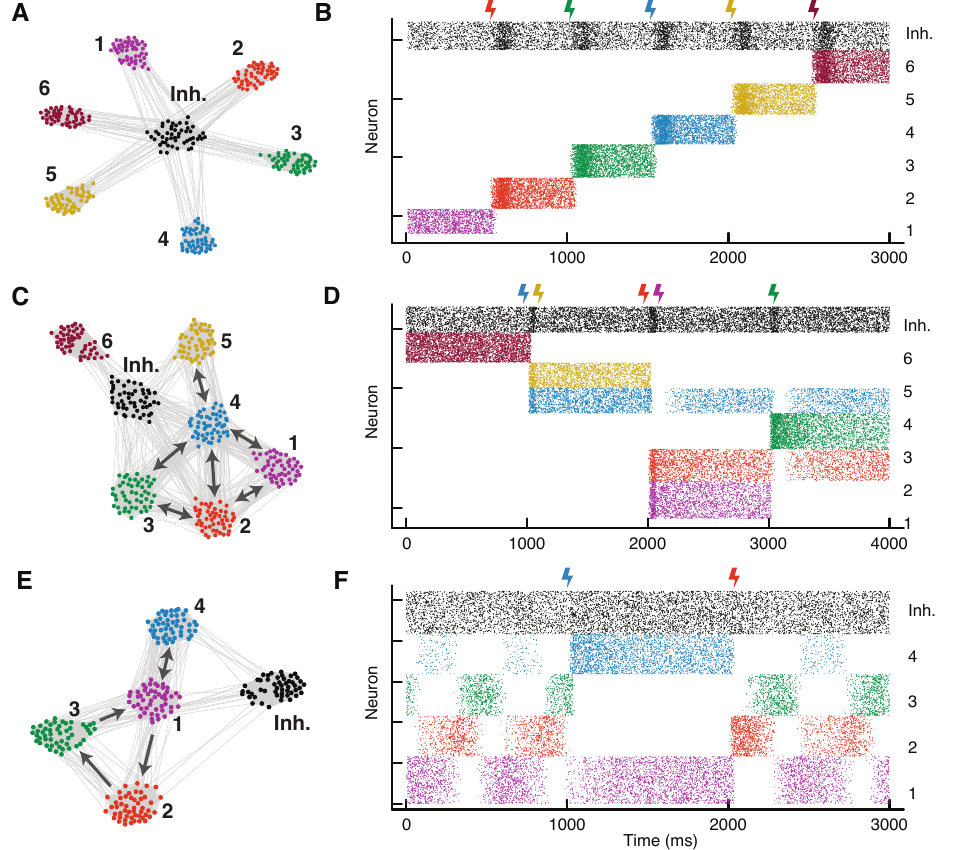}
    \caption{(A) Clustered network with six excitatory clusters (color) and one inhibitory cluster (black).  No connections between clusters. (B) Raster plot of activity of clustered network in panel A, for excitatory cluster size $N_E = NI = 1000$, $\tau_E = 40 \unit{ms}$, $\tau_I = 20\unit{ms}$. Stimulus of duration $10 \unit{ms}$ applied to each neuron in clusters 2 (red),   3 (green),  4 (blue),   5 (yellow), and  6 (maroon) at times $t = 500, 1000, 1500, 2000, 2500$ respectively (lightning bolts). (C) Clustered network with six excitatory clusters (color) and one inhibitory cluster (black).  Pattern of bi-directional connectivity  between excitatory clusters indicated with arrows. (D)  Raster plot of activity of clustered network from panel C for same parameters as B. Stimulus of duration $10 \unit{ms}$ applied to each neuron in clusters 4 and 5 (yellow and blue), 1 and 2 (red and purple), and 3 (green) at times $t = 1000, 2000, 3000$ respectively (lightning bolts).  (E) Clustered network with four excitatory clusters (color) and one inhibitory cluster (black).  Pattern of directed connectivity  between excitatory clusters indicated with arrows. (F) Raster plot of activity of clustered network from panel C for same parameters as B. Stimulus of duration $10 \unit{ms}$ applied to each neuron in clusters 4  (blue),  2 (red), at times $t = 1000, 2000$ respectively (lightning bolts).   }
    \label{fig:diverse_cluster_dynamics}
\end{figure*}

We consider a stochastic spiking network with $n$ excitatory clusters of $N_E$ neurons each and $N_I$ inhibitory neurons (e.g., Fig.~\ref{fig:diverse_cluster_dynamics}A). 
 Within-cluster connectivity is stronger than between-cluster connectivity and the inhibitory neurons connect non-specifically to all excitatory clusters.
 Connectivity is sparse and random and synaptic weights scale inversely with the population size.
This connectivity determines a neuron-by-neuron weight matrix $J_{ij}$. Excitatory and inhibitory neurons have distinct synaptic time constants  $\tau_E $ and $\tau_I$ and receive distinct external inputs $b_E$ and $b_I.$

Each neuron's spike train is a point process $\dot s_i(t)$ whose time dependent rate depends on the voltage $v_i(t)$. The joint dynamics of the voltage and spike train are given by 
\begin{equation}
\label{eqn:v_model_spk}
\begin{aligned}
    \tau_{\alpha(i)} \dot v_{i} &= - v_i + \sum_{j =1}^N J_{ij} \dot s_j + b_{\alpha(i)} ,\\
     \dot s_i &\sim \mathcal{PP} \big(\lfloor v_i\rfloor_+\big).
     \end{aligned}
\end{equation}
Here, $\lfloor x \rfloor_+ := \max(x,0)$. We denote the population containing neuron $i$ as $\alpha(i)$, so that $\tau_{\alpha(i)} = \tau_E$ or $\tau_I$ and $b_{\alpha(i)} = b_E$ or $b_I$. Additional model details can be found in Appendix \ref{sec:model_details}. The spike trains $\dot s_i$ are sums of delta functions, where $\dot s_i(t) = \sum_{t'\in T_i} \delta(t-t')$ and $T_i$ is the set of spike times for neuron $i$.  Equation \ref{eqn:v_model_spk} can alternately be written in integral form as 
\begin{align*}
    v_i(t) = v_i(0) + b_{\alpha(i)} + \frac{1}{\tau_{\alpha(i)}}\int_0^t e^{(t' - t)/\tau_{\alpha(i)}} \dot s_j(t') \, dt' .
\end{align*}
This is a nonlinear Hawkes process with excitatory and inhibitory interactions. 

When there is no connectivity between the clusters, or when the connectivity between the clusters is weak and unstructured, the network has $n$ metastable states with a unique active excitatory cluster each (Fig.~\ref{fig:diverse_cluster_dynamics}B). 
Pulses of external input can drive transitions between attractors; spontaneous noise-driven transitions are also possible over a longer timescale.

More varied dynamics become possible when there is structured connectivity between clusters. We describe the structure of connectivity between clusters with a directed graph $G$: An edge from cluster $\alpha$ to cluster $\beta$ in $G$ means that the connections from cluster $\alpha$ to cluster $\beta$ are stronger than the baseline level of connectivity between clusters. 
When the pattern of connectivity between clusters is symmetric (Fig.~\ref{fig:diverse_cluster_dynamics}C), the network still displays metastability between fixed points (Fig.~\ref{fig:diverse_cluster_dynamics}D). However, multiple excitatory clusters may be active in one attractor and individual clusters may participate in multiple attractors. When the pattern of connectivity between clusters is not symmetric (Fig.~\ref{fig:diverse_cluster_dynamics}E), the network can have dynamic attractors such as periodic orbits or chaotic attractors, in addition to metastable fixed points. For example, the network illustrated in Fig.~\ref{fig:diverse_cluster_dynamics}E displays metastability between a periodic orbit and a fixed point (Fig.~\ref{fig:diverse_cluster_dynamics}F).

\subsection{Mean-field theory for clustered, inhibition-stabilized networks is dynamically equivalent to the CTLN model}

\begin{figure*}[ht!]
    \centering
    \includegraphics[]{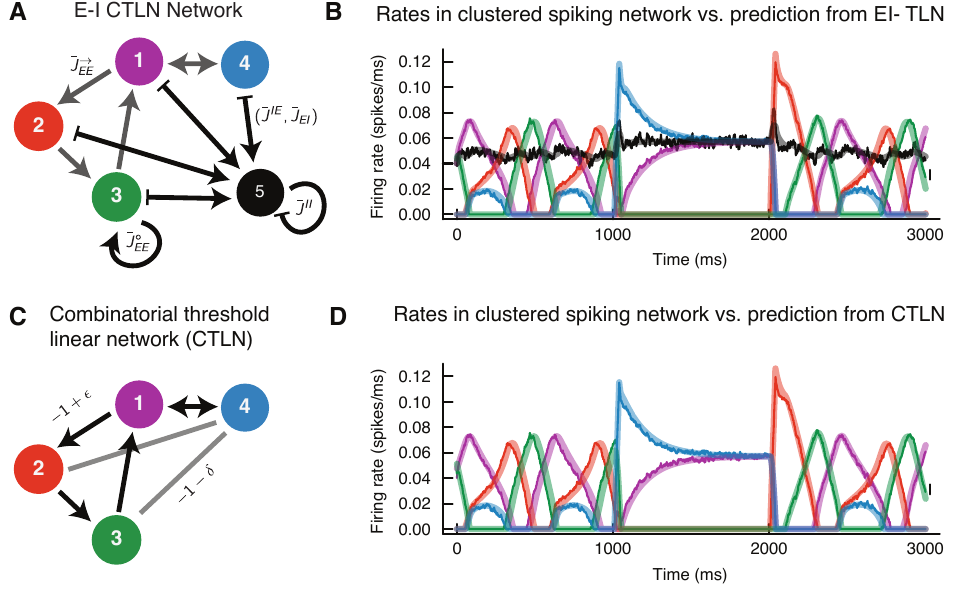}
    \caption{(A) EI-TLN corresponding to the clustered network in Fig.~\ref{fig:diverse_cluster_dynamics} E.  Gray edges represent the directed pattern of connectivity between excitatory clusters, while black edges represent connections to and from the inhibitory cluster.  (B)
     Mean firing rates in clustered spiking network (thick, transparent line) versus prediction from EI-TLN (thin, solid line) with external input as described in  Fig.~\ref{fig:diverse_cluster_dynamics} F.  (C) CTLN corresponding to the clustered network in panel A.  Black edges represent the directed pattern of connectivity between excitatory clusters, which are now effectively weakly inhibitory, with weight $-1+\epsilon$. Gray edges  between non-connected clusters represent strongly inhibitory effective connectivity, with weight $-1 -\delta$.   (D)  Firing rates in clustered spiking network (thick, transparent line) versus prediction from CTLN (thin, solid line) with external input as described in  Fig.~\ref{fig:diverse_cluster_dynamics} F.  }
    \label{fig:mean_field}
\end{figure*}


In the limit of large network size ($N_E, N_I \to \infty$), the cluster-mean firing rates obey $n+1$ dimensional deterministic dynamics. We will show that in the additional limit of $\tau_I/\tau_E \to 0$, this rate model is dynamically equivalent to the CTLN model. 
Let  $\bar J_{\alpha\beta}$ denote the mean weight of synapses from neurons in cluster $\beta$ to a neuron in cluster $\alpha$. The connection weights in our microscopic model depend only on the excitatory/inhibitory identity of the clusters and the directed pattern of connectivity between clusters (Fig.~\ref{fig:mean_field}A). So, the expected value of the mean weight $\bar J_{\alpha\beta}$ takes one of six possible values:
\begin{align}
\bar J_{\alpha \beta} = \begin{cases} 
\bar J^\circlearrowleft_{EE} &\mbox { if } \beta = \alpha, \alpha, \beta  \mbox{ excitatory}\\
\bar J_{EE}^{\to}  &\mbox { if } \beta \to \alpha \in G,  \alpha, \beta  \mbox{ excitatory}\\
\bar J_{EE}^{\not\to}  &\mbox { if } \beta \not\to \alpha \in G,  \alpha, \beta  \mbox{ excitatory}\\
\bar J_{EI} &\mbox { if } \alpha  \mbox{ excitatory}, \beta \mbox{ inhibitory }\\
\bar J_{IE} &\mbox { if } \alpha \mbox{ inhibitory } , \beta \mbox{ excitatory}\\
\bar J_{II} &\mbox { if } \alpha, \beta \mbox{ inhibitory }.
\end{cases}
\label{eqn:weights}
\end{align} 
Now, we take the expectation over connectivity, spike train realization, and the identity of a neuron within its cluster in  \ref{eqn:v_model_spk}  over. Letting  $v_{\alpha} = \abk{v_i}_{i\in \alpha}$, we find that in the large-$N$ limit the expected cluster-mean voltages evolve as 
\begin{align}
    \tau_\alpha \dot v_{\alpha}= -v_{\alpha}  + \sum_{\beta=1}^n \bar J_{\alpha\beta} \lfloor v_\beta \rfloor_+ + b_\alpha. \label{eqn:v_model_rate}
\end{align}
Corresponding mean-field limits for similar nonlinear Hawkes models are proven in \cite{stiefel_mean-field_2023, pfaffelhuber_mean-field_2022, zhu_large_2015, chevallier_mean-field_2017, delattre_hawkes_2016, heesen_fluctuation_2021}, in particular for the multi-population case in \cite{ditlevsen_multi-class_2017}.


Our goal is now to relate the dynamics of the mean firing rates to the graph structure $G$. To do this, we will show that, under certain conditions, the excitatory firing rates in the EI-TLN model are approximated by those of the combinatorial threshold linear network (CTLN) model introduced in \cite{curto_fixed_2019}. 
The CTLN model is highly mathematically tractable, exposing elegant links between network structure and dynamics. By grounding CTLNs in a model with explicit excitatory and inhibitory populations and spikes, we can apply results about CTLNs to understand a broader class of models. 
Although units in the CTLN model are traditionally interpreted as single neurons, we interpret them as populations.  Under this interpretation, the CTLN model represents the connectivity between clusters of excitatory neurons, over a background of fast shared inhibition, via a pattern of weak and strong inhibition. 

The CTLN model is defined by a rule for translating a graph into a weighted adjacency matrix. Specifically, if the pattern of connectivity between clusters is given by a directed graph $G$, then the weights in the CTLN model are given by
 \begin{align}
    W_{\alpha\beta} =  \begin{cases}
         -1 + \epsilon &\mbox { if }  \beta \to \alpha \in G \\
           -1 - \delta &\mbox { if }  \beta \to \alpha \notin G,\\
         \end{cases} \label{eqn:ctln_weights}
 \end{align}
 where $0 < \epsilon < 1$ and $ \delta > 0$ (Fig.~\ref{fig:mean_field} C). Additionally, we require $(-1 +\epsilon)(-1-\delta)>1$; see \cite{curto_fixed_2019} for details.  The nodes' dynamics are 
  \begin{align} \label{eqn:ctln}
     \dot x_\alpha = -x_\alpha + \left \lfloor \sum_{\beta=1}^n W_{\alpha \beta} x_\beta + b\right \rfloor _+,
 \end{align}
 where $x_\alpha$ is the mean firing rate in population $\alpha$ and $b > 0$ is an constant, positive external input which is the same for all populations. 

The CTLN model, Eqs. \ref{eqn:ctln_weights},\ref{eqn:ctln}, differs from the mean-field voltage dynamics, Eq. \ref{eqn:v_model_rate}, in two ways. First, while the nonlinearity is inside of the summation in the mean-field voltage dynamics, it is outside of the summation in the CTLN dynamics. Second, while the mean-field voltage dynamics has both positive and negative interaction weights, all weights in the CTLN are negative. We will thus relate the mean-field theory for the clustered spiking network to the CTLN in two steps.

First, we address the location of the nonlinearity using the results of Miller and Fumarola in \cite{miller_mathematical_2012}. Miller and Fumarola show that for any system that evolves according to $\tau \frac{dv}{dt} = -v + J f(v) + b(t)$, $v\in \R^d$, there is an equivalent system which evolves according to $\tau \frac{dx}{dt} = -x + f(J x + c(t))$, where the external inputs $c$ to the $x$-model is a low-pass-filtered version of the input to the $v$-model. When $J$ is invertible, the mapping between systems is a topological conjugacy, and even outside this case, the mapping between the $x$-system and the $v$- system is given by a linear map. This means that the attractor landscapes of the two models are qualitatively the same.

Thus, to predict the mean-field dynamics of our clustered spiking networks, we can alternately study the model
   \begin{align} \label{eqn:ei_tln}
   \tau_\alpha  \dot x_\alpha = -x_\alpha + \left \lfloor \sum_{\beta=1}^n \bar J_{\alpha \beta} x_\beta + c_{\alpha}\right \rfloor _+,
 \end{align}
where $x_{\alpha}$ is the mean firing rate in population $\alpha$. We term systems of this form the excitatory-inhibitory threshold linear network (EI-TLN) model. Although this model has the same functional form as the CTLN model, it differs in that the weight matrix $J$ still retains the explicit excitatory and inhibitory weights from the clustered spiking network. In all following equations, we consider the constant-input case, and thus write $ b_{\alpha}$ instead of $c_{\alpha}$ for the input. 
Figure~\ref{fig:mean_field}B shows an example of the population-averaged firing rates $\abk{f(v_i)}_{i\in \alpha}$ in the clustered spiking network against predicted firing mean rates $f(v_\alpha) = f(J x_\alpha + c_{\alpha})$ generated by simulating Eq. \ref{eqn:ei_tln}, with matching initial conditions and parameters. Details can be found in Appendix \ref{sec:trans}. 


To relate the EI-TLN model (Eqs.~\ref{eqn:weights},~\ref{eqn:ei_tln}) to the CTLN model (Eqs.~\ref{eqn:ctln_weights},~\ref{eqn:ctln}) we consider the limit of fast inhibition: $\tau_I /\tau_E \to 0$. 
This separation of timescales implies that the mean firing rate of the inhibitory cluster is directly proportional to the mean firing rate across all excitatory clusters:
\begin{align*}
    x_{I} \approx \frac{ 1}{1 - \bar J_{II}} \left(\bar J_{IE} \sum_{\beta \neq I} x_\beta + b_I\right).
\end{align*}
Inserting this into Eq.~\ref{eqn:ei_tln} results in effective inhibitory connections between excitatory clusters and allows us to derive conditions on the weights of the EI-TLN model, $\bar J$, that describe when it limits to the CTLN model.  The first condition is that within-cluster excitation exactly cancels inhibition:
\begin{align} \label{eqn:ei_balance}
    \bar J_{EE}^\circlearrowleft + \frac{\bar J_{EI} \bar J_{IE}}{1 - \bar J_{II}} = 0.
\end{align}
When this condition is met, we can relate the CTLN parameters $\epsilon$ and $
\delta$ to the EI-TLN parameters as 
\begin{align*}
   \epsilon =& \bar J_{EE}^\to  -  \bar J_{EE}^\circlearrowleft +1\\ 
   \delta =& -  \left(\bar J_{EE}^{\not\to} -  \bar J_{EE}^\circlearrowleft + 1\right) 
   \\
   \mbox{ and } b =& \frac{\bar J_{EI}}{1 - \bar J_{II}} b_I +  b_E.
\end{align*}
This allows us to interpret the CTLN constraints $0 < \epsilon < 1$ and $0 < \delta$ as  constraints on the relative strength of within cluster and between-cluster excitation:

\begin{equation}
\label{eqn:circ_vs_notto}
    \begin{aligned}
   0 >\bar J_{EE}^\to -  \bar J_{EE}^\circlearrowleft > -1\\
\bar J_{EE}^{\not\to}  - \bar J_{EE}^\circlearrowleft <  -1 \\
   \left( \bar J_{EE}^\to -  \bar J_{EE}^\circlearrowleft \right) \big(\bar J_{EE}^{\not\to} - \bar J_{EE}^\circlearrowleft \big) > 1.
\end{aligned}
\end{equation}

When these three additional conditions are met, the EI-TLN reduces to the CTLN.
These imply $\bar{J}_{EE}^\circlearrowleft > 1$, so the network is always in an inhibition-stabilized regime; we will explore this in Section \ref{sec:paradox}. 

Figure ~\ref{fig:mean_field}D shows an example of the CTLN firing rates against the cluster-averaged rates of excitatory neurons, for relatively fast inhibition ($\tau_I/\tau_E = 0.5$). We see that the CTLN rates match the average rates as well as the EI-TLN rates, even with this finite time scale. In Section \ref{sec:fixed_point_bif}, we discuss how the dynamics in the EI-TLN model vary as we change the ratio between the inhibitory and excitatory time scales. Details of this derivation are given in Appendix \ref{sec:ei_tln_to_ctln}. 

Relating clustered spiking networks to the CTLN model in this way allows us to apply results about CTLNs to study the dynamics of clustered spiking networks. In the next two sections, we will survey how results about CTLNs let us construct clustered spiking networks with diverse dynamics that depend on the directed graph structure of between-cluster connectivity. 

\subsection{Graph rules for CTLNs predict dynamics of clustered spiking networks.}
\label{sec:graph_rules}

\subsubsection{Metastable fixed points in clustered spiking networks correspond to maximal cliques in the cluster graph.}

\label{sec:fixed_point}
Figures~\ref{fig:diverse_cluster_dynamics}A and C show two examples of clustered spiking networks that exhibit multiple metastable states. The first network, without strong connectivity between clusters, has a metastable state corresponding to each excitatory cluster, where neurons from that cluster are active while neurons in other excitatory  clusters are silent. The second network, where there is strong, reciprocal connectivity between some pairs of clusters and not others, has multiple metastable states where a subset of clusters are active. 
We apply results about stable fixed points of CTLNs to relate these metastable states to the graph $G$ describing the connectivity between clusters.

Because CTLNs are threshold-linear, they have at most one fixed point in each linear chamber of their phase space.  (Except potentially in degenerate cases, i.e., when the linear system in a chamber has a zero eigenvalue.). Within a chamber, it is easy to recover the fixed point location via matrix inversion. Thus, the fixed point structure of a CTLN is fully described by the set of fixed point supports, which are defined as 
 $\supp (\vec x^*)= \{i \mid x^*_i > 0\}.$  We denote this set of fixed point supports as 
\begin{align*}
    \FP (G, \epsilon, \delta) = \{\supp(\vec x^*) \mid \vec x^* \mbox{a fixed point}  \}
\end{align*}
 Due to the constraint $b > 0,$ the origin is never a fixed point, so the empty set is never a fixed point support.  
Because the CTLN model is a mean-field theory for our clustered spiking model, fixed points in clustered spiking networks are characterized by the clusters they are supported on. Because the evolution of the firing rates in our clustered spiking network is stochastic, stable fixed points in the deterministic CTLN model correspond to metastable fixed points in the clustered spiking network for finite cluster size.

When the graph $G$ is symmetric ( $i\to j$ if and only if $j\to i$) convergence to a fixed point is guaranteed \cite{hahnloser_digital_2000, curto_pattern_2016}. Thus, characterizing the stable fixed points provides a complete picture of the attractor landscape of symmetric networks. Stable fixed points of symmetric CTLNs correspond to \emph{maximal cliques}\footnote{A clique is a set of vertices $\sigma\subseteq[n]$ such that $i \leftrightarrow j$ for all $i, j\in  \sigma$. A clique is \emph{maximal} if it is not contained in any larger clique. ($[n]$ is the set of integers from $1$ to $n$.)} of $G$ \cite{curto_fixed_2019, curto_pattern_2016}.  In Fig.~\ref{fig:fixed_points}A,  
we highlight the maximal cliques in the example from Fig.~\ref{fig:diverse_cluster_dynamics}C. Note that ``trivial" examples of maximal cliques, such as isolated vertices or bidirectional edges that are not contained within larger cliques, still support stable fixed points. The number of maximal cliques that a graph can have grows exponentially in the number of vertices, so the number of potential stable fixed points in a CTLN also grows exponentially with $n$ \cite{curto_stable_2024, moon_cliques_1965}. 

This characterization of the stable fixed point is partially extended to non-symmetric graphs in \cite{curto_fixed_2019} and \cite{morrison_diversity_2024}. A node $k\in [n]$ is a target for a clique $\sigma \subseteq [n]$ if $ j\to k$ for each $j\in \sigma.$ By Theorem 3.5 of \cite{morrison_diversity_2024}, a clique supports a stable fixed point if and only if it is \emph{target-free}.  This is conjectured to completely characterize the stable fixed points of CTLNs; there are no known examples of stable fixed points of CTLNs whose support is not a  target-free clique. This conjecture can be guaranteed to hold under some additional assumptions; for details see  \cite{curto_fixed_2019}. The metastable fixed point of the clustered spiking network in Fig.~\ref{fig:diverse_cluster_dynamics}E is supported on the target-free clique $\{1,4\}$. 
\begin{figure}
    \centering
    \includegraphics{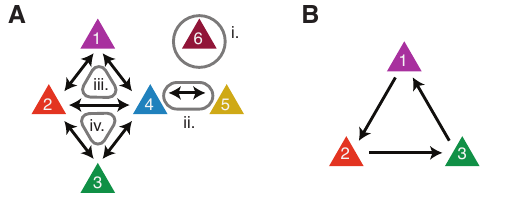}
    \caption{A. Maximal cliques in the connectivity graph of the clustered spiking network pictured in (Fig.~\ref{fig:diverse_cluster_dynamics} C.) Notice that each maximal clique corresponds to a fixed point in (Fig.~\ref{fig:diverse_cluster_dynamics} D).  B. Directed 3-cycle.  }
    \label{fig:fixed_points}
\end{figure}

\subsubsection{When between-cluster connectivity is not symmetric, clustered spiking networks have dynamic attractors with bounded total activity.}

While symmetric connectivity guarantees convergence to a stable fixed point in CTLNs, an alternate condition guarantees the absence of stable fixed points. If a graph is 1) oriented (has no bidirectional edges) and 2) has no sinks (every vertex has at least one outgoing edge), then the CTLN on that graph has no stable fixed points. 
This does not, however, mean that the activity explodes. 

The total activity of CTLNs is bounded above and below:
for any CTLN defined by a graph $G$ and parameters $0 < \epsilon < 1, \; \delta >0$, and any initial condition $\vec x (0)$, there exists a time $T$ such that for all $t > T$
\begin{align}
    \frac{b}{1 +\delta} < \sum_{\alpha = 1}^n x_\alpha(t) < \frac{b}{1 -\epsilon}.
    \label{eqn:ctln_total_bound}
\end{align}
In other words, all trajectories in CTLNs eventually approach a region where the total population activity is bounded \cite{lienkaemper_combinatorial_2022}. 

Because CTLNs have bounded total activity, networks without stable fixed points have dynamic attractors \cite{morrison_diversity_2024}. 
The simplest example of an oriented graph with no sinks is a directed 3-cycle (Fig.~\ref{fig:fixed_points}B). This CTLN has a limit cycle attractor, established in \cite{bel_periodic_2021}. Although the weights in the CTLN model are all negative, activity still often follows edges in the graph. Graph rules can also be used to predict the dynamic attractors of CTLNs from the graph structure \cite{parmelee_core_2022, parmelee_sequential_2022}. 

Notice that if $\epsilon$ and $\delta$ are held constant in Eq.~\ref{eqn:ctln_total_bound}, the total activity bound does not scale with the number of clusters. Thus, in networks where many clusters participate in an attractor, the mean firing rates in each cluster are proportionately lower, scaling as $1/n$. This can be avoided by scaling $\epsilon$ and $\delta$ with the number of clusters. 
In the EI-TLN, in the limit where $\tau_I\to 0$, this result translates to 
\begin{align}
    \frac{\frac{\bar J_{EI}}{1 - \bar J_{II}} b_I +  b_E}{  \bar J_{EE}^\circlearrowleft - \bar J_{EE}^{\not\to}} < \sum_{\alpha = 1}^n x_\alpha(t) <    \frac{\frac{\bar J_{EI}}{1 - \bar J_{II}} b_I +  b_E}{  \bar J_{EE}^\circlearrowleft  - \bar J_{EE}^\to } . 
    \label{eqn:act_bount_ei}
\end{align}
To keep the mean firing rates constant in $n$ thus requires either making between-cluster excitation stronger or within-cluster excitation weaker. In Appendix \ref{sec:total_bound}, we prove that this bound  also holds for finite, sufficiently small, inhibitory timescales. 

\subsection{Clustered spiking networks exhibit paradoxical effects }
\label{sec:paradox}
\begin{figure}[ht!]
    \centering
    \includegraphics{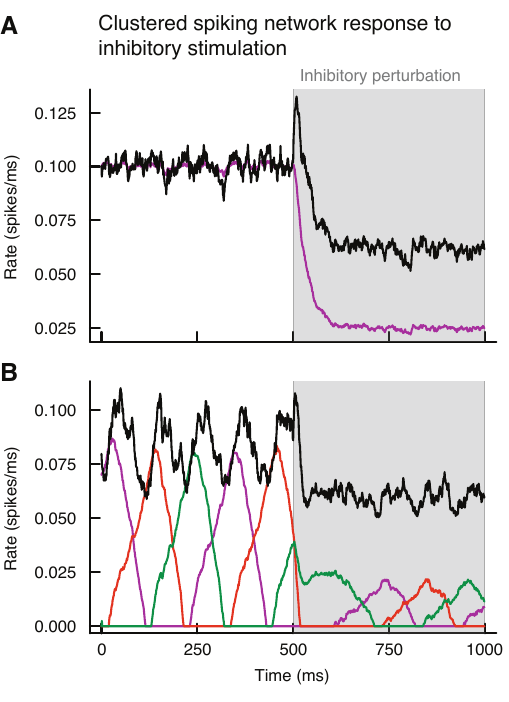}
    \caption{Paradoxical effect in clustered spiking networks. Black: average inhibitory firing rate. Color: average rates within each excitatory cluster. (A) Network with a single excitatory cluster.  From $T = 500 \unit{ms}$ onward, excitatory stimulation of $.1 \unit {spikes/ms}$ is applied to the inhibitory population only. Notice that the firing rates of the inhibitory neurons, as well as the excitatory neurons, decrease. (B) A network with three excitatory clusters, connected in a directed cycle.  The clustered spiking network has a periodic orbit. From $T = 500 \unit{ms}$ onward, excitatory stimulation of  $.1 \unit {spikes/ms}$  is applied to the inhibitory population only. Again, the firing rates of the inhibitory neurons, as well as the excitatory neurons, decrease.  }
    \label{fig:para}
\end{figure}

We can interpret this bound as a result about the stability of the total mean excitatory rate, even in the absence of stable fixed points. This stability is provided by the inhibitory population. Notice that since $J_{EE}^{\not\to} \geq 0$, Eq. \ref{eqn:circ_vs_notto} implies $J_{EE}^\circlearrowleft \geq 1.$ This means that the excitatory-only subnetwork is not stable, so the EI-TLN is in an inhibition-stabilized regime. As an inhibition-stabilized network (ISN), the EI-TLN displays the so-called paradoxical effect: stimulation of the inhibitory population suppresses the excitatory clusters, leading in turn to a decrease in the inhibitory firing rates \cite{sadeh_inhibitory_2021, tsodyks_paradoxical_1997}. This effect is considered a signature of inhibitory stabilization, and is observed in many cortical areas across species \cite{sadeh_inhibitory_2021}. 

 While the paradoxical effect in ISNs is typically described as describing the effect of perturbations when the network is at a fixed point (Fig.~\ref{fig:para}A), we observe the paradoxical effect of inhibitory stimulation on dynamic attractors of the clustered spiking network as well (Fig.~\ref{fig:para}B).  For sufficiently fast inhibition, we prove in Appendix \ref{sec:app_paradox} that stimulating the inhibitory population reduces inhibitory firing rates along all dynamic attractors in EI-TLNs.  More precisely, changing the external inputs $b_E, b_I$ scales the rates along attractors of the EI-TLN.  For trajectories obeying Eq.~\ref{eqn:act_bount_ei}, increasing $b_I$ scales down both the excitatory and inhibitory rates. Because this is a linear transformation, and because solutions of the $v$-model and the $x$- model must obey the relationship $f(v) = x +\tau \frac{dx}{dt}$, this phenomenon is also seen in solutions of the $v$-model (Eq. \ref{eqn:v_model_rate}) and thus the clustered spiking network (Eq. \ref{eqn:v_model_spk}).

\subsection{When inhibition is slow, a synchronized oscillation emerges}

\begin{figure*}[ht!]
\centering
    \includegraphics{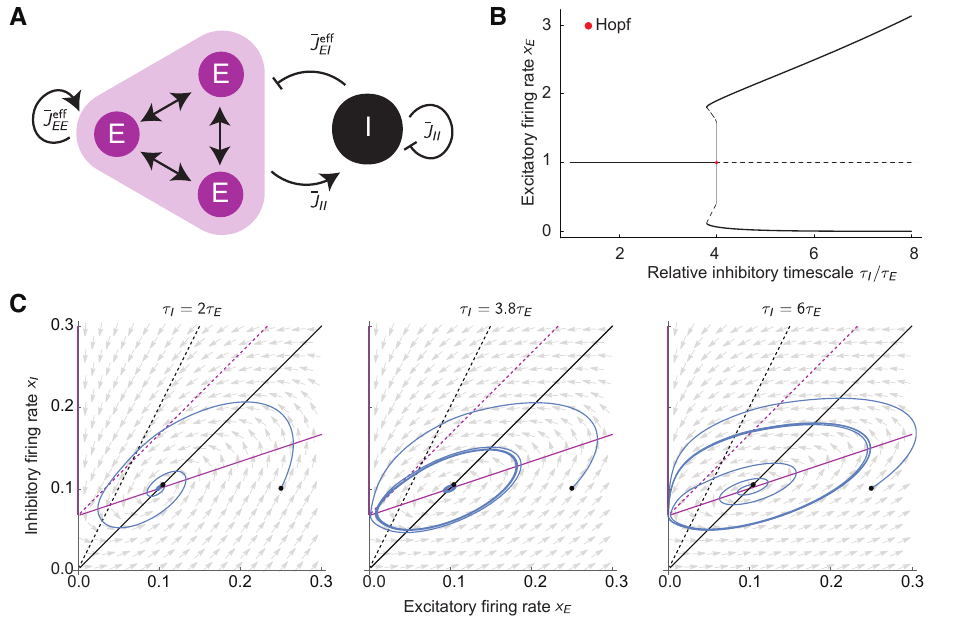}
    \caption{(A) EI-TLN network with one excitatory maximal clique.  (B) Bifurcation diagram. Solid lines indicate stability, dashed indicates instability. Hopf bifurcation point marked in red.  Minimum and maximum values of $x_E$ along limit cycle are shown.   (C) Phase plots for $\tau_I/\tau_E = 2$ (left), 3.8 (center), 6 (right). Gray arrows indicate vector field. Nullclines (solid) and linear chamber boundaries (dashed) for $x_E$ and $x_I$ are shown in purple and black respectively. Trajectories for two initial conditions ($(x_E, x_I )= (.1, .25)$ and $(x_E, x_I) = (.105, .105)$, black circles) are shown in blue.  Left panel is from region with stable fixed point only, center is from bistable region, right is  from region with stable limit cycle and unstable fixed point. Parameters: $\bar J_{EE}^\circlearrowleft = 1.5, \, \bar J_{EE}^\to = .75, \, \bar J_{EI} = -1.5, \, \bar J_{IE} = 2, \, \bar J_{II} = -1$.}
    \label{fig:fp_bif}
\end{figure*}

We have shown that the CTLN model is a mean-field theory for clustered spiking networks for large cluster size and fast inhibition. However, excitatory and inhibitory membrane time constants are typically both approximately 10-30 ms. We thus relax the assumption of fast inhibition, comparing the dynamics of the EI-TLN to the CTLN model as we vary the ratio between the inhibitory and excitatory timescales. 

\subsubsection{Oscillation emerges from stable fixed points via non-smooth Hopf bifurcation}
As we observed in Section \ref{sec:fixed_point}, there is a one-to-one correspondence between fixed points in the CTLN and the EI-TLN. We first ask how fast inhibition needs to be for a stable fixed point of the CTLN to remain stable in the EI-TLN. All known fixed points in CTLNs are supported on target-free cliques, so we restrict our attention to stable fixed points of this type. 

We fix the parameters $\bar J_{EE}, \bar J_{EI}, \bar J_{IE}, \bar J_{II},$ and $\tau_E$. Additionally, we consider the size $k$ of the clique as a parameter. We consider the linearized dynamics around this fixed point, which is the linear system in the chamber containing the fixed point. At the fixed point, all other excitatory clusters in the network are below threshold; we thus neglect them in this analysis. 

All clusters within the clique are, by definition, connected identically to one another. So, we can describe the dynamics via a two-dimensional linear dynamical system describing the mean excitatory and inhibitory firing rates, $x_E$ and $x_I$  (Fig.~\ref{fig:fp_bif}A): 
\begin{align*}
  \tau_E \, \dot x_E &= -x_E + \bar J_{EE}^\mathrm{eff} \, x_E  + \bar J_{EI} \, x_I + b_E, \\
    \tau_I \, \dot x_I &= -x_I + k \, \bar J_{IE} \, \bar J_{II} \, x_I \, x_E +b_I,
\end{align*}
where $\bar J_{EE}^\mathrm{eff} = (k-1) \bar J_{EE}^\to +\bar J_{EE}^\circlearrowleft $ is the total excitatory-to-excitatory connectivity. 
This linear system has a single fixed point, which is stable when 
\begin{align*}
    \frac{\tau_I}{\tau_E } < -\frac{\bar J_{II} - 1}{\bar J_{EE}^{\mathrm{eff}} -1}.
\end{align*}
Notice that the strength of inhibition relative to excitation needed to stabilize the network for a fixed ratio of timescales increases with the number of clusters in the clique.
The EI-TLN dynamics are non-smooth, so the bifurcation corresponding to this loss of stability is not locally described by a standard topological normal form. 
When the fixed point loses stability, the two eigenvalues of the linearization about the fixed point always have nonzero imaginary parts, thus solutions spiral away from the fixed point when it loses stability. 

We investigate this bifurcation through numerical continuation in a network with one excitatory maximal clique (Fig. \ref{fig:fp_bif}A). In this case, the critical value is $\tau_I/\tau_E = 4$ (Fig. \ref{fig:fp_bif}B). The bifurcation gives rise to a large-amplitude periodic orbit along which all excitatory neurons are synchronized, reminiscent of a pyramidal-interneuron gamma  oscillation \cite{whittington_inhibition-based_2000}. There is a region of bistability between the fixed point and a large-amplitude limit cycle for $3.8 \leq \tau_I/\tau_E < 4.$ Within the bistable region, the stable limit cycle and the fixed point are separated by an unstable periodic orbit. So, we call the bifurcation a subcritical non-smooth Hopf bifurcation. The unstable limit cycle is born with finite amplitude because the fixed point is in the interior of a linear chamber, which cannot support a periodic orbit except exactly at the bifurcation point (Fig.~\ref{fig:fp_bif}C). Details are in Section \ref{sec:fixed_point_bif}. A version of this calculation also appears in \cite{tsodyks_paradoxical_1997}, where a limit cycle emerges when inhibition is made weaker.



\subsubsection{Multistability between synchronized oscillation and CTLN-like limit cycle}

 \begin{figure*}[ht!]
    \centering
    \includegraphics{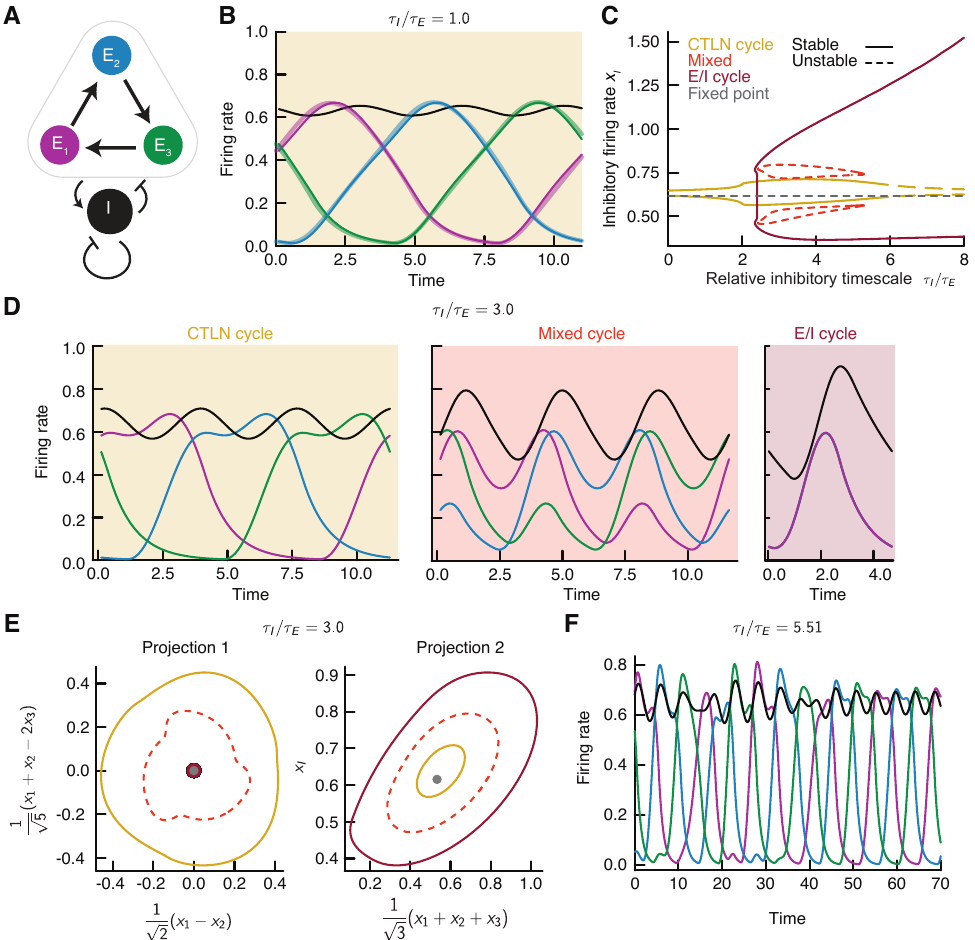}
    \caption{(A) E-I TLN whose underlying graph is a directed 3-cycle. (B) Solid: firing rates of the EI-TLN for $\tau_I/\tau_E = 1.0$. Transparent lines: firing rates in CTLN model with same graph.  (C) Bifurcation diagram for the parameter $\tau_I/\tau_E$. Solid lines indicate stability, dashed indicates instability. Periodic orbits are represented by the minimum and maximum values of the inhibitory firing rate along the orbit.  (D) Three periodic orbits present for $\tau_I/\tau_E = 3.0$. Left, the CTLN cycle, obtained by continuation of stable periodic orbit in panel B. Notice that the shapes of the firing rate curves are deformed relative to those in panel B. Center, the unstable periodic orbit. Right, the E/I cycle, a stable periodic orbit in which all three excitatory neurons have synchronous activity, and the total exciatory and inhibitory populations oscillate. Analogous to the limit cycle described in Fig.~\ref{fig:fp_bif}. (E) Two projections of all three periodic orbits for $\tau_I/\tau_E = 3.0.$ Gold is CTLN cycle, dark red is E/I cycle, dashed orange is unstable periodic orbit separating the other two cycles. (F) Firing rates for a trajectory which follows a long, apparently quasiperiodic transient before approaching the CTLN limit cycle. Parameters: $ J_{EE}^\circlearrowleft = 1.5, \, J_{EE}^\to = 0.75, \, J_{EE}^{\not\to} = 0, \, J_{EI} = -2.25, \, J_{IE} = 2, \, J_{II} = -2, \, \tau_E=1$.}
    \label{fig:3cycle_bif}
\end{figure*}

Next, we extend this analysis beyond fixed points by considering the simplest example of a CTLN with a nontrivial dynamic attractor. This is the CTLN whose graph is a directed three-cycle (Fig.~\ref{fig:3cycle_bif}A), which has a single unstable fixed point and a stable limit cycle for all allowed values of $\epsilon, \delta$ \cite{bel_periodic_2021}.  At the fixed point, the linearized system has one negative eigenvalue and two complex-conjugate eigenvalues with positive real part. Initial conditions near the fixed point spiral outward to the limit cycle.  Because the dynamics of the excitatory populations in the EI-TLN are described by the CTLN equation in the limit where $\tau_I \to 0$, the EI-TLN also has an unstable fixed point and a stable limit cycle when $\tau_I$ is small relative to $\tau_E$. Using numerical continuation from this limit cycle, we investigate how the attractors of the EI-TLN change with $\tau_I$.


 When $\tau_I/\tau_E =1$, the EI-TLN has a CTLN-like limit cycle (Fig.~\ref{fig:3cycle_bif}B).  
 We numerically continued this limit cycle, which we call the ``CTLN cycle", for $\tau_E = 1,\, 0.01 \leq \tau_I \leq 8.0$  (Fig.~\ref{fig:3cycle_bif}C, gold). This periodic orbit remains stable up to $\tau_I/\tau_E = $ 6.46, though its shape becomes distorted as $\tau_I$ increases (Fig.~\ref{fig:3cycle_bif}D, left, shown for $\tau_I/\tau_E = 3$). We continued this periodic orbit past its instability and did not find an endpoint. 

 At $\tau_I/\tau_E = 2.4$ , the second pair of complex-conjugate eigenvalues of the fixed point crosses the real axis and the fixed point no longer has any stable direction. At this parameter value, a second limit cycle we call the ``EI-cycle" branches off from the fixed point via a subcritical non-smooth Hopf-like bifurcation (Fig.~\ref{fig:3cycle_bif}C, dark red) reminiscent of the one described in the previous section. Along this limit cycle, all three excitatory neurons have synchronous activity (Fig.~\ref{fig:3cycle_bif}D, right). There is a region of bistability between the CTLN cycle and the E-I cycle. 

 For parameters in a large part of the bistable region, from $2.42 \leq \tau_I/\tau_E \leq 5.29$, there are two unstable periodic orbits between the CTLN cycle and the EI cycle, which we call the ``mixed cycles" (Fig.~\ref{fig:3cycle_bif}C, orange dashed; Fig.~\ref{fig:3cycle_bif}D, middle). For $\tau_I/\tau_E = 3.0$, we plot two projections of all three identified periodic orbits in Fig.~\ref{fig:3cycle_bif}E. The mixed cycle is annihilated via a non-smooth fold bifurcation with itself at $\tau_I/\tau_E = 5.29$, which is before the CTLN cycle loses stability. 
 
 Thus, there is a region of bistability between the CTLN cycle and the EI-cycle where we have not been able to find an unstable perioidic orbit separating them. Because the CTLN limit cycle loses stability via a secondary Hopf bifurcation, we expect an unstable quasiperiodic trajectory/unstable invariant torus in this region. In Fig.~\ref{fig:3cycle_bif}F, we plot a trajectory that appears to stay near that invariant torus for some time before eventually entering the CTLN cycle. Details are available in Section \ref{sec:limit_cycle_bif}.

\section{Discussion}
We derived a rate model dynamically equivalent to the CTLN model as a mean-field theory for cluster-average firing rates in an inhibition-stabilized, clustered nonlinear Hawkes model in the limit of large cluster size and fast inhibition. This allows us to predict diverse dynamics of clustered networks using existing results proved for CTLNs. We showed that the CTLN model describes the mean firing rates of excitatory populations in an inhibition-stabilized regime. This allowed us to demonstrate the paradoxical effect, a hallmark of inhibitory stabilization, in networks with multiple excitatory populations and diverse dynamics. Finally, while the equivalence between the CTLN and the clustered spiking networks holds in the limit of fast inhibition, we showed that CTLNs qualitatively predict the dynamics of clustered spiking networks for more reasonable inhibitory timescales.

Many other derivations of rate models from models with a higher level of biological detail are possible--several are reviewed in \cite{ermentrout_firing_2010}. These derivations can be done rigorously in the limit where synapses are slow \cite{ermentrout_reduction_1994, rinzel_activity_1992}, or for large populations with asynchronous activity \cite{shriki_rate_2003}. Our choice of a rate model of the form of Equation \ref{eqn:v_model_spk},
with the summation inside of the nonlinearity, reflects the assumption that the synaptic decay time constant is fast relative to the membrane time constant, while the alternate form assumes that the synaptic decay time is slow relative to the membrane time constant \cite{ermentrout_firing_2010}. Both forms assume that the synaptic kernel's rise time is fast relative to both the synaptic decay time and the membrane time constant \cite{pinto_spatially_2001}. Stochastic spiking models, like the Hawkes process we used here, can also be related to more biologically detailed models \cite{aviel_spiking_2006, ostojic_computational_2024}. Many of these derivations make it possible to relate parameters in the rate model to biophysically meaningful values.  Thus, by connecting the CTLN model to networks that incorporate explicit excitatory and inhibitory populations, as well as spikes, we make it possible to use techniques from the theory of CTLNs to predict the qualitative behavior of more detailed models.

In particular, we can now use \emph{graph rules} that relate a network's structure to its dynamics, proven for CTLNs, in a broader context. We briefly review some of these results here---a more in-depth review can be found in \cite{curto_graph_2023}.  Many graph rules relate whether a subnetwork (in our case a set of clusters) can be co-active at a fixed point to both the structure of the subnetwork and the way the subnetwork is embedded in the full network. For instance, the result we stated in Section \ref{sec:graph_rules}, that a clique supports a stable fixed point if and only if it is target-free, is a graph rule of this type \cite{curto_stable_2024}. A second class of graph rules, including \emph{gluing rules} and \emph{nerve theorems}, predicts the fixed points of networks that are constructed compositionally, i.e., out of smaller networks connected in constrained ways \cite{curto_fixed_2019, santander_nerve_2022}.  

Finally, dynamic attractors of CTLNs can also be partially predicted from graph structure. The simplest example of a CTLN which has a dynamic attractor, the directed 3-cycle, is proven to have a periodic orbit in \cite{bel_periodic_2021}.  Even in this simple case, the proof is nontrivial. Consequently, current methods for predicting dynamic attractors from network structure more generally are heuristic. Dynamic attractors of CTLNs are typically observed to correspond to unstable \emph{core motifs}, which are minimal subnetworks that support an unstable fixed point \cite{parmelee_core_2022}.  When these networks have a cycle-like structure, this cycle-like structure predicts the order in which neurons fire in the attractor. 
By combining these graph rules, it is possible to construct recurrent networks that perform biologically relevant tasks, such as locomotion, by hand rather than by training \cite{londono_alvarez_juliana_attractor-based_2024,parmelee_sequential_2022}. This results in networks with a modular structure that are amenable to mathematical analysis.  By relating the CTLN model to a model that incorporates spiking activity and explicit inhibitory and excitatory populations, we help to interpret these results about CTLNs in a biological context---for instance, we have shown that it makes more sense to interpret ``neurons" in the CTLN model as populations.

While our derivation of the CTLN model from clustered spiking networks was in the limit where inhibition is fast relative to excitation, we have observed that the qualitative dynamics of CTLNs match the qualitative dynamics of clustered spiking networks for much more reasonable excitatory and inhibitory timescales, including in cases where inhibitory time constants are larger than excitatory time constants. We gave explicit conditions on the ratio between the excitatory and inhibitory timescales required so that fixed points that are stable in the CTLN remain stable in the EI-TLN. Because these conditions depend on the weights of the EI-TLN, and because there are multiple combinations of EI-TLN weights which give rise to the same set of CTLN weights, it is possible to scale the EI-TLN weights in order to stabilize CTLN fixed points for any combination of excitatory and inhibitory time constants. We also note that whether a stable fixed point in the CTLN remains stable in the EI-TLN depends on the number of clusters which are active at the fixed point. By scaling the probability of between-cluster connectivity with network size, we can also ensure that CTLN fixed points remain stable with fixed weights. 

Due to our choice of $1/N$ weight scaling, the clustered spiking network model here converge to deterministic mean-field dynamics in the limit of large $N$. For smaller $N$, this model exhibits finite-size fluctuations. Thus, stable fixed points, limit cycles, and other attractors in the CTLN model become metastable attractors in finite-size clustered networks.
Metastable fixed points in similar networks model various phenomena in the brain \cite{mazzucato_expectation-induced_2019, mazzucato_stimuli_2016, brinkman_metastable_2022, rossi_dynamical_2024}. 
Finite-size fluctuations can shape firing rates, trigger oscillations, and drive transitions between metastable attractors~\cite{schwalger_mind_2019}.
Friedlin-Wentzell theory provides a well-established mathematical framework to estimate transition times between metastable fixed points when fluctuations are small, as for large but finite $N$ \cite{freidlin_random_2012, bressloff_path_2014, brinkman_metastable_2022}. Clustered nonlinear Hawkes networks also exhibit metastability between fixed points and dynamic attractors, as in Fig. \ref{fig:diverse_cluster_dynamics}E. Noise-induced transitions involving dynamic attractors can be estimated using recent extensions of Friedlin-Wentzell theory \cite{fleurantin_dynamical_2023}.

Networks with strong coupling take synaptic weights to scale as $1/\sqrt{N}$ rather than $1/N$~\cite{barral_synaptic_2016}. This heterogeneity generates activity fluctuations that are uncorrelated in the large $N$ limit \cite{sompolinsky_chaos_1988}. This result has been extended to various observed features of neural connectivity, such as block-structured, low-rank, spatial, or correlated connectivity and heavy-tailed synaptic weight distributions \cite{zeraati_topology-dependent_2024, mastrogiuseppe_linking_2018, dahmen_strong_2023, hu_motif_2013, hu_spectrum_2022, aljadeff_transition_2015, stern_dynamics_2014, clark_structure_2025, wardak_extended_2022}. 
Our reduction of clustered networks to the CTLN averages over the activity within each cluster, suggesting that a similar reduction might be possible in networks with strong coupling.


\section*{Acknowledgments}

C.L. is supported by the Swartz Foundation and the BU Center for Systems Neuroscience.  The authors benefited from helpful discussions with Carina Curto and Ryan Goh. 

\bibliographystyle{apsrev4-2}
\bibliography{references}

\begin{thebibliography}{71}%
\makeatletter
\providecommand \@ifxundefined [1]{%
 \@ifx{#1\undefined}
}%
\providecommand \@ifnum [1]{%
 \ifnum #1\expandafter \@firstoftwo
 \else \expandafter \@secondoftwo
 \fi
}%
\providecommand \@ifx [1]{%
 \ifx #1\expandafter \@firstoftwo
 \else \expandafter \@secondoftwo
 \fi
}%
\providecommand \natexlab [1]{#1}%
\providecommand \enquote  [1]{``#1''}%
\providecommand \bibnamefont  [1]{#1}%
\providecommand \bibfnamefont [1]{#1}%
\providecommand \citenamefont [1]{#1}%
\providecommand \href@noop [0]{\@secondoftwo}%
\providecommand \href [0]{\begingroup \@sanitize@url \@href}%
\providecommand \@href[1]{\@@startlink{#1}\@@href}%
\providecommand \@@href[1]{\endgroup#1\@@endlink}%
\providecommand \@sanitize@url [0]{\catcode `\\12\catcode `\$12\catcode
  `\&12\catcode `\#12\catcode `\^12\catcode `\_12\catcode `\%12\relax}%
\providecommand \@@startlink[1]{}%
\providecommand \@@endlink[0]{}%
\providecommand \url  [0]{\begingroup\@sanitize@url \@url }%
\providecommand \@url [1]{\endgroup\@href {#1}{\urlprefix }}%
\providecommand \urlprefix  [0]{URL }%
\providecommand \Eprint [0]{\href }%
\providecommand \doibase [0]{https://doi.org/}%
\providecommand \selectlanguage [0]{\@gobble}%
\providecommand \bibinfo  [0]{\@secondoftwo}%
\providecommand \bibfield  [0]{\@secondoftwo}%
\providecommand \translation [1]{[#1]}%
\providecommand \BibitemOpen [0]{}%
\providecommand \bibitemStop [0]{}%
\providecommand \bibitemNoStop [0]{.\EOS\space}%
\providecommand \EOS [0]{\spacefactor3000\relax}%
\providecommand \BibitemShut  [1]{\csname bibitem#1\endcsname}%
\let\auto@bib@innerbib\@empty
\bibitem [{\citenamefont {Ocker}\ \emph {et~al.}(2017)\citenamefont {Ocker},
  \citenamefont {Hu}, \citenamefont {Buice}, \citenamefont {Doiron},
  \citenamefont {Josić}, \citenamefont {Rosenbaum},\ and\ \citenamefont
  {Shea-Brown}}]{ocker_statistics_2017}%
  \BibitemOpen
  \bibfield  {author} {\bibinfo {author} {\bibfnamefont {G.~K.}\ \bibnamefont
  {Ocker}}, \bibinfo {author} {\bibfnamefont {Y.}~\bibnamefont {Hu}}, \bibinfo
  {author} {\bibfnamefont {M.~A.}\ \bibnamefont {Buice}}, \bibinfo {author}
  {\bibfnamefont {B.}~\bibnamefont {Doiron}}, \bibinfo {author} {\bibfnamefont
  {K.}~\bibnamefont {Josić}}, \bibinfo {author} {\bibfnamefont
  {R.}~\bibnamefont {Rosenbaum}},\ and\ \bibinfo {author} {\bibfnamefont
  {E.}~\bibnamefont {Shea-Brown}},\ }\href
  {https://doi.org/10.1016/j.conb.2017.07.011} {\bibfield  {journal} {\bibinfo
  {journal} {Current Opinion in Neurobiology}\ }\bibinfo {series}
  {Computational {Neuroscience}},\ \textbf {\bibinfo {volume} {46}},\ \bibinfo
  {pages} {109} (\bibinfo {year} {2017})}\BibitemShut {NoStop}%
\bibitem [{\citenamefont {La~Camera}(2022)}]{la_camera_mean_2022}%
  \BibitemOpen
  \bibfield  {author} {\bibinfo {author} {\bibfnamefont {G.}~\bibnamefont
  {La~Camera}},\ }in\ \href {https://doi.org/10.1007/978-3-030-89439-9_6}
  {{\selectlanguage {English}\emph {\bibinfo {booktitle} {Computational
  {Modelling} of the {Brain}: {Modelling} {Approaches} to {Cells}, {Circuits}
  and {Networks}}}}},\ \bibinfo {editor} {edited by\ \bibinfo {editor}
  {\bibfnamefont {M.}~\bibnamefont {Giugliano}}, \bibinfo {editor}
  {\bibfnamefont {M.}~\bibnamefont {Negrello}},\ and\ \bibinfo {editor}
  {\bibfnamefont {D.}~\bibnamefont {Linaro}}}\ (\bibinfo  {publisher} {Springer
  International Publishing},\ \bibinfo {address} {Cham},\ \bibinfo {year}
  {2022})\ pp.\ \bibinfo {pages} {125--157}\BibitemShut {NoStop}%
\bibitem [{\citenamefont {Galves}\ \emph {et~al.}(2024)\citenamefont {Galves},
  \citenamefont {Löcherbach},\ and\ \citenamefont
  {Pouzat}}]{galves_probabilistic_2024}%
  \BibitemOpen
  \bibfield  {author} {\bibinfo {author} {\bibfnamefont {A.}~\bibnamefont
  {Galves}}, \bibinfo {author} {\bibfnamefont {E.}~\bibnamefont
  {Löcherbach}},\ and\ \bibinfo {author} {\bibfnamefont {C.}~\bibnamefont
  {Pouzat}},\ }\href {https://doi.org/10.1007/978-3-031-68409-8}
  {{\selectlanguage {English}\emph {\bibinfo {title} {Probabilistic {Spiking}
  {Neuronal} {Nets}: {Neuromathematics} for the {Computer} {Era}}}}},\ Lecture
  {Notes} on {Mathematical} {Modelling} in the {Life} {Sciences}\ (\bibinfo
  {publisher} {Springer International Publishing},\ \bibinfo {address} {Cham},\
  \bibinfo {year} {2024})\BibitemShut {NoStop}%
\bibitem [{\citenamefont {Markram}(1997)}]{markram_network_1997}%
  \BibitemOpen
  \bibfield  {author} {\bibinfo {author} {\bibfnamefont {H.}~\bibnamefont
  {Markram}},\ }\href {https://doi.org/10.1093/cercor/7.6.523} {\bibfield
  {journal} {\bibinfo  {journal} {Cerebral Cortex}\ }\textbf {\bibinfo {volume}
  {7}},\ \bibinfo {pages} {523} (\bibinfo {year} {1997})}\BibitemShut {NoStop}%
\bibitem [{\citenamefont {Perin}\ \emph {et~al.}(2011)\citenamefont {Perin},
  \citenamefont {Berger},\ and\ \citenamefont {Markram}}]{perin_synaptic_2011}%
  \BibitemOpen
  \bibfield  {author} {\bibinfo {author} {\bibfnamefont {R.}~\bibnamefont
  {Perin}}, \bibinfo {author} {\bibfnamefont {T.~K.}\ \bibnamefont {Berger}},\
  and\ \bibinfo {author} {\bibfnamefont {H.}~\bibnamefont {Markram}},\ }\href
  {https://doi.org/10.1073/pnas.1016051108} {\bibfield  {journal} {\bibinfo
  {journal} {Proceedings of the National Academy of Sciences}\ }\textbf
  {\bibinfo {volume} {108}},\ \bibinfo {pages} {5419} (\bibinfo {year}
  {2011})},\ \bibinfo {note} {publisher: Proceedings of the National Academy of
  Sciences}\BibitemShut {NoStop}%
\bibitem [{\citenamefont {Yoshimura}\ \emph {et~al.}(2005)\citenamefont
  {Yoshimura}, \citenamefont {Dantzker},\ and\ \citenamefont
  {Callaway}}]{yoshimura_excitatory_2005}%
  \BibitemOpen
  \bibfield  {author} {\bibinfo {author} {\bibfnamefont {Y.}~\bibnamefont
  {Yoshimura}}, \bibinfo {author} {\bibfnamefont {J.~L.~M.}\ \bibnamefont
  {Dantzker}},\ and\ \bibinfo {author} {\bibfnamefont {E.~M.}\ \bibnamefont
  {Callaway}},\ }\href {https://doi.org/10.1038/nature03252} {\bibfield
  {journal} {\bibinfo  {journal} {Nature}\ }\textbf {\bibinfo {volume} {433}},\
  \bibinfo {pages} {868} (\bibinfo {year} {2005})},\ \bibinfo {note}
  {publisher: Nature Publishing Group}\BibitemShut {NoStop}%
\bibitem [{\citenamefont {Litwin-Kumar}\ and\ \citenamefont
  {Doiron}(2014)}]{litwin-kumar_formation_2014}%
  \BibitemOpen
  \bibfield  {author} {\bibinfo {author} {\bibfnamefont {A.}~\bibnamefont
  {Litwin-Kumar}}\ and\ \bibinfo {author} {\bibfnamefont {B.}~\bibnamefont
  {Doiron}},\ }\href {https://doi.org/10.1038/ncomms6319} {\bibfield  {journal}
  {\bibinfo  {journal} {Nature Communications}\ }\textbf {\bibinfo {volume}
  {5}},\ \bibinfo {pages} {5319} (\bibinfo {year} {2014})},\ \bibinfo {note}
  {publisher: Nature Publishing Group}\BibitemShut {NoStop}%
\bibitem [{\citenamefont {Ocker}\ and\ \citenamefont
  {Doiron}(2019)}]{ocker_training_2019}%
  \BibitemOpen
  \bibfield  {author} {\bibinfo {author} {\bibfnamefont {G.~K.}\ \bibnamefont
  {Ocker}}\ and\ \bibinfo {author} {\bibfnamefont {B.}~\bibnamefont {Doiron}},\
  }\href {https://doi.org/10.1093/cercor/bhy001} {\bibfield  {journal}
  {\bibinfo  {journal} {Cerebral Cortex}\ }\textbf {\bibinfo {volume} {29}},\
  \bibinfo {pages} {937} (\bibinfo {year} {2019})}\BibitemShut {NoStop}%
\bibitem [{\citenamefont {Zenke}\ \emph {et~al.}(2015)\citenamefont {Zenke},
  \citenamefont {Agnes},\ and\ \citenamefont {Gerstner}}]{zenke_diverse_2015}%
  \BibitemOpen
  \bibfield  {author} {\bibinfo {author} {\bibfnamefont {F.}~\bibnamefont
  {Zenke}}, \bibinfo {author} {\bibfnamefont {E.~J.}\ \bibnamefont {Agnes}},\
  and\ \bibinfo {author} {\bibfnamefont {W.}~\bibnamefont {Gerstner}},\ }\href
  {https://doi.org/10.1038/ncomms7922} {\bibfield  {journal} {\bibinfo
  {journal} {Nature Communications}\ }\textbf {\bibinfo {volume} {6}},\
  \bibinfo {pages} {6922} (\bibinfo {year} {2015})},\ \bibinfo {note}
  {publisher: Nature Publishing Group}\BibitemShut {NoStop}%
\bibitem [{\citenamefont {Montangie}\ \emph {et~al.}(2020)\citenamefont
  {Montangie}, \citenamefont {Miehl},\ and\ \citenamefont
  {Gjorgjieva}}]{montangie_autonomous_2020}%
  \BibitemOpen
  \bibfield  {author} {\bibinfo {author} {\bibfnamefont {L.}~\bibnamefont
  {Montangie}}, \bibinfo {author} {\bibfnamefont {C.}~\bibnamefont {Miehl}},\
  and\ \bibinfo {author} {\bibfnamefont {J.}~\bibnamefont {Gjorgjieva}},\
  }\href {https://doi.org/10.1371/journal.pcbi.1007835} {\bibfield  {journal}
  {\bibinfo  {journal} {PLOS Computational Biology}\ }\textbf {\bibinfo
  {volume} {16}},\ \bibinfo {pages} {e1007835} (\bibinfo {year} {2020})},\
  \bibinfo {note} {publisher: Public Library of Science}\BibitemShut {NoStop}%
\bibitem [{\citenamefont {Brinkman}\ \emph {et~al.}(2022)\citenamefont
  {Brinkman}, \citenamefont {Yan}, \citenamefont {Maffei}, \citenamefont
  {Park}, \citenamefont {Fontanini}, \citenamefont {Wang},\ and\ \citenamefont
  {La~Camera}}]{brinkman_metastable_2022}%
  \BibitemOpen
  \bibfield  {author} {\bibinfo {author} {\bibfnamefont {B.~A.~W.}\
  \bibnamefont {Brinkman}}, \bibinfo {author} {\bibfnamefont {H.}~\bibnamefont
  {Yan}}, \bibinfo {author} {\bibfnamefont {A.}~\bibnamefont {Maffei}},
  \bibinfo {author} {\bibfnamefont {I.~M.}\ \bibnamefont {Park}}, \bibinfo
  {author} {\bibfnamefont {A.}~\bibnamefont {Fontanini}}, \bibinfo {author}
  {\bibfnamefont {J.}~\bibnamefont {Wang}},\ and\ \bibinfo {author}
  {\bibfnamefont {G.}~\bibnamefont {La~Camera}},\ }\href
  {https://doi.org/10.1063/5.0062603} {\bibfield  {journal} {\bibinfo
  {journal} {Applied Physics Reviews}\ }\textbf {\bibinfo {volume} {9}},\
  \bibinfo {pages} {011313} (\bibinfo {year} {2022})}\BibitemShut {NoStop}%
\bibitem [{\citenamefont {Litwin-Kumar}\ and\ \citenamefont
  {Doiron}(2012)}]{litwin-kumar_slow_2012}%
  \BibitemOpen
  \bibfield  {author} {\bibinfo {author} {\bibfnamefont {A.}~\bibnamefont
  {Litwin-Kumar}}\ and\ \bibinfo {author} {\bibfnamefont {B.}~\bibnamefont
  {Doiron}},\ }\href {https://doi.org/10.1038/nn.3220} {\bibfield  {journal}
  {\bibinfo  {journal} {Nature Neuroscience}\ }\textbf {\bibinfo {volume}
  {15}},\ \bibinfo {pages} {1498} (\bibinfo {year} {2012})},\ \bibinfo {note}
  {publisher: Nature Publishing Group}\BibitemShut {NoStop}%
\bibitem [{\citenamefont {Rossi}\ \emph {et~al.}(2024)\citenamefont {Rossi},
  \citenamefont {Budzinski}, \citenamefont {Medeiros}, \citenamefont
  {Boaretto}, \citenamefont {Muller},\ and\ \citenamefont
  {Feudel}}]{rossi_dynamical_2024}%
  \BibitemOpen
  \bibfield  {author} {\bibinfo {author} {\bibfnamefont {K.~L.}\ \bibnamefont
  {Rossi}}, \bibinfo {author} {\bibfnamefont {R.~C.}\ \bibnamefont
  {Budzinski}}, \bibinfo {author} {\bibfnamefont {E.~S.}\ \bibnamefont
  {Medeiros}}, \bibinfo {author} {\bibfnamefont {B.~R.~R.}\ \bibnamefont
  {Boaretto}}, \bibinfo {author} {\bibfnamefont {L.}~\bibnamefont {Muller}},\
  and\ \bibinfo {author} {\bibfnamefont {U.}~\bibnamefont {Feudel}},\ }\href
  {http://arxiv.org/abs/2305.05328} {{\selectlanguage {English}\bibinfo {title}
  {Dynamical properties and mechanisms of metastability: a perspective in
  neuroscience}}} (\bibinfo {year} {2024}),\ \bibinfo {note} {arXiv:2305.05328
  [q-bio]}\BibitemShut {NoStop}%
\bibitem [{\citenamefont {Mazzucato}\ \emph {et~al.}(2019)\citenamefont
  {Mazzucato}, \citenamefont {La~Camera},\ and\ \citenamefont
  {Fontanini}}]{mazzucato_expectation-induced_2019}%
  \BibitemOpen
  \bibfield  {author} {\bibinfo {author} {\bibfnamefont {L.}~\bibnamefont
  {Mazzucato}}, \bibinfo {author} {\bibfnamefont {G.}~\bibnamefont
  {La~Camera}},\ and\ \bibinfo {author} {\bibfnamefont {A.}~\bibnamefont
  {Fontanini}},\ }\href {https://doi.org/10.1038/s41593-019-0364-9} {\bibfield
  {journal} {\bibinfo  {journal} {Nature Neuroscience}\ }\textbf {\bibinfo
  {volume} {22}},\ \bibinfo {pages} {787} (\bibinfo {year} {2019})},\ \bibinfo
  {note} {publisher: Nature Publishing Group}\BibitemShut {NoStop}%
\bibitem [{\citenamefont {Amit}(1997)}]{amit_model_1997}%
  \BibitemOpen
  \bibfield  {author} {\bibinfo {author} {\bibfnamefont {D.}~\bibnamefont
  {Amit}},\ }\href {https://doi.org/10.1093/cercor/7.3.237} {\bibfield
  {journal} {\bibinfo  {journal} {Cerebral Cortex}\ }\textbf {\bibinfo {volume}
  {7}},\ \bibinfo {pages} {237} (\bibinfo {year} {1997})}\BibitemShut {NoStop}%
\bibitem [{\citenamefont {Wang}(2008)}]{wang_decision_2008}%
  \BibitemOpen
  \bibfield  {author} {\bibinfo {author} {\bibfnamefont {X.-J.}\ \bibnamefont
  {Wang}},\ }\href {https://doi.org/10.1016/j.neuron.2008.09.034} {\bibfield
  {journal} {\bibinfo  {journal} {Neuron}\ }\textbf {\bibinfo {volume} {60}},\
  \bibinfo {pages} {215} (\bibinfo {year} {2008})}\BibitemShut {NoStop}%
\bibitem [{\citenamefont {Moreno-Bote}\ \emph {et~al.}(2007)\citenamefont
  {Moreno-Bote}, \citenamefont {Rinzel},\ and\ \citenamefont
  {Rubin}}]{moreno-bote_noise-induced_2007}%
  \BibitemOpen
  \bibfield  {author} {\bibinfo {author} {\bibfnamefont {R.}~\bibnamefont
  {Moreno-Bote}}, \bibinfo {author} {\bibfnamefont {J.}~\bibnamefont
  {Rinzel}},\ and\ \bibinfo {author} {\bibfnamefont {N.}~\bibnamefont
  {Rubin}},\ }\href {https://doi.org/10.1152/jn.00116.2007} {\bibfield
  {journal} {\bibinfo  {journal} {Journal of Neurophysiology}\ }\textbf
  {\bibinfo {volume} {98}},\ \bibinfo {pages} {1125} (\bibinfo {year}
  {2007})},\ \bibinfo {note} {publisher: American Physiological
  Society}\BibitemShut {NoStop}%
\bibitem [{\citenamefont {Albantakis}\ and\ \citenamefont
  {Deco}(2011)}]{albantakis_changes_2011}%
  \BibitemOpen
  \bibfield  {author} {\bibinfo {author} {\bibfnamefont {L.}~\bibnamefont
  {Albantakis}}\ and\ \bibinfo {author} {\bibfnamefont {G.}~\bibnamefont
  {Deco}},\ }\href {https://doi.org/10.1371/journal.pcbi.1002086} {\bibfield
  {journal} {\bibinfo  {journal} {PLOS Computational Biology}\ }\textbf
  {\bibinfo {volume} {7}},\ \bibinfo {pages} {e1002086} (\bibinfo {year}
  {2011})},\ \bibinfo {note} {publisher: Public Library of Science}\BibitemShut
  {NoStop}%
\bibitem [{\citenamefont {Marder}\ and\ \citenamefont
  {Bucher}(2001)}]{marder_central_2001}%
  \BibitemOpen
  \bibfield  {author} {\bibinfo {author} {\bibfnamefont {E.}~\bibnamefont
  {Marder}}\ and\ \bibinfo {author} {\bibfnamefont {D.}~\bibnamefont
  {Bucher}},\ }\href {https://doi.org/10.1016/s0960-9822(01)00581-4} {\bibfield
   {journal} {\bibinfo  {journal} {Current biology: CB}\ }\textbf {\bibinfo
  {volume} {11}},\ \bibinfo {pages} {R986} (\bibinfo {year}
  {2001})}\BibitemShut {NoStop}%
\bibitem [{\citenamefont {Curto}\ and\ \citenamefont
  {Morrison}(2023)}]{curto_graph_2023}%
  \BibitemOpen
  \bibfield  {author} {\bibinfo {author} {\bibfnamefont {C.}~\bibnamefont
  {Curto}}\ and\ \bibinfo {author} {\bibfnamefont {K.}~\bibnamefont
  {Morrison}},\ }\href {https://doi.org/10.1090/noti2661} {\bibfield  {journal}
  {\bibinfo  {journal} {Notices of the American Mathematical Society}\ }\textbf
  {\bibinfo {volume} {70}},\ \bibinfo {pages} {1} (\bibinfo {year}
  {2023})}\BibitemShut {NoStop}%
\bibitem [{\citenamefont {Curto}\ and\ \citenamefont
  {Morrison}(2016)}]{curto_pattern_2016}%
  \BibitemOpen
  \bibfield  {author} {\bibinfo {author} {\bibfnamefont {C.}~\bibnamefont
  {Curto}}\ and\ \bibinfo {author} {\bibfnamefont {K.}~\bibnamefont
  {Morrison}},\ }\href {https://doi.org/10.1162/NECO_a_00869} {\bibfield
  {journal} {\bibinfo  {journal} {Neural Computation}\ }\textbf {\bibinfo
  {volume} {28}},\ \bibinfo {pages} {2825} (\bibinfo {year}
  {2016})}\BibitemShut {NoStop}%
\bibitem [{\citenamefont {Curto}\ \emph {et~al.}(2019)\citenamefont {Curto},
  \citenamefont {Geneson},\ and\ \citenamefont {Morrison}}]{curto_fixed_2019}%
  \BibitemOpen
  \bibfield  {author} {\bibinfo {author} {\bibfnamefont {C.}~\bibnamefont
  {Curto}}, \bibinfo {author} {\bibfnamefont {J.}~\bibnamefont {Geneson}},\
  and\ \bibinfo {author} {\bibfnamefont {K.}~\bibnamefont {Morrison}},\ }\href
  {https://doi.org/10.1162/neco_a_01151} {\bibfield  {journal} {\bibinfo
  {journal} {Neural Computation}\ }\textbf {\bibinfo {volume} {31}},\ \bibinfo
  {pages} {94} (\bibinfo {year} {2019})}\BibitemShut {NoStop}%
\bibitem [{\citenamefont {Morrison}\ \emph {et~al.}(2024)\citenamefont
  {Morrison}, \citenamefont {Degeratu}, \citenamefont {Itskov},\ and\
  \citenamefont {Curto}}]{morrison_diversity_2024}%
  \BibitemOpen
  \bibfield  {author} {\bibinfo {author} {\bibfnamefont {K.}~\bibnamefont
  {Morrison}}, \bibinfo {author} {\bibfnamefont {A.}~\bibnamefont {Degeratu}},
  \bibinfo {author} {\bibfnamefont {V.}~\bibnamefont {Itskov}},\ and\ \bibinfo
  {author} {\bibfnamefont {C.}~\bibnamefont {Curto}},\ }\href
  {https://doi.org/10.1137/22M1541666} {\bibfield  {journal} {\bibinfo
  {journal} {SIAM Journal on Applied Dynamical Systems}\ }\textbf {\bibinfo
  {volume} {23}},\ \bibinfo {pages} {855} (\bibinfo {year} {2024})}\BibitemShut
  {NoStop}%
\bibitem [{\citenamefont {Curto}\ \emph {et~al.}(2024)\citenamefont {Curto},
  \citenamefont {Geneson},\ and\ \citenamefont {Morrison}}]{curto_stable_2024}%
  \BibitemOpen
  \bibfield  {author} {\bibinfo {author} {\bibfnamefont {C.}~\bibnamefont
  {Curto}}, \bibinfo {author} {\bibfnamefont {J.}~\bibnamefont {Geneson}},\
  and\ \bibinfo {author} {\bibfnamefont {K.}~\bibnamefont {Morrison}},\ }\href
  {https://doi.org/10.1016/j.aam.2023.102652} {\bibfield  {journal} {\bibinfo
  {journal} {Advances in Applied Mathematics}\ }\textbf {\bibinfo {volume}
  {154}},\ \bibinfo {pages} {102652} (\bibinfo {year} {2024})}\BibitemShut
  {NoStop}%
\bibitem [{\citenamefont {Parmelee}\ \emph
  {et~al.}(2022{\natexlab{a}})\citenamefont {Parmelee}, \citenamefont
  {Alvarez}, \citenamefont {Curto},\ and\ \citenamefont
  {Morrison}}]{parmelee_sequential_2022}%
  \BibitemOpen
  \bibfield  {author} {\bibinfo {author} {\bibfnamefont {C.}~\bibnamefont
  {Parmelee}}, \bibinfo {author} {\bibfnamefont {J.~L.}\ \bibnamefont
  {Alvarez}}, \bibinfo {author} {\bibfnamefont {C.}~\bibnamefont {Curto}},\
  and\ \bibinfo {author} {\bibfnamefont {K.}~\bibnamefont {Morrison}},\ }\href
  {https://doi.org/10.1137/21M1445120} {\bibfield  {journal} {\bibinfo
  {journal} {SIAM Journal on Applied Dynamical Systems}\ }\textbf {\bibinfo
  {volume} {21}},\ \bibinfo {pages} {1597} (\bibinfo {year}
  {2022}{\natexlab{a}})},\ \bibinfo {note} {publisher: Society for Industrial
  and Applied Mathematics}\BibitemShut {NoStop}%
\bibitem [{\citenamefont {Parmelee}\ \emph
  {et~al.}(2022{\natexlab{b}})\citenamefont {Parmelee}, \citenamefont {Moore},
  \citenamefont {Morrison},\ and\ \citenamefont {Curto}}]{parmelee_core_2022}%
  \BibitemOpen
  \bibfield  {author} {\bibinfo {author} {\bibfnamefont {C.}~\bibnamefont
  {Parmelee}}, \bibinfo {author} {\bibfnamefont {S.}~\bibnamefont {Moore}},
  \bibinfo {author} {\bibfnamefont {K.}~\bibnamefont {Morrison}},\ and\
  \bibinfo {author} {\bibfnamefont {C.}~\bibnamefont {Curto}},\ }\href
  {https://doi.org/10.1371/journal.pone.0264456} {\bibfield  {journal}
  {\bibinfo  {journal} {PLOS ONE}\ }\textbf {\bibinfo {volume} {17}},\ \bibinfo
  {pages} {e0264456} (\bibinfo {year} {2022}{\natexlab{b}})},\ \bibinfo {note}
  {publisher: Public Library of Science}\BibitemShut {NoStop}%
\bibitem [{\citenamefont {Santander}\ \emph {et~al.}(2022)\citenamefont
  {Santander}, \citenamefont {Ebli}, \citenamefont {Patania}, \citenamefont
  {Sanderson}, \citenamefont {Burtscher}, \citenamefont {Morrison},\ and\
  \citenamefont {Curto}}]{santander_nerve_2022}%
  \BibitemOpen
  \bibfield  {author} {\bibinfo {author} {\bibfnamefont {D.~E.}\ \bibnamefont
  {Santander}}, \bibinfo {author} {\bibfnamefont {S.}~\bibnamefont {Ebli}},
  \bibinfo {author} {\bibfnamefont {A.}~\bibnamefont {Patania}}, \bibinfo
  {author} {\bibfnamefont {N.}~\bibnamefont {Sanderson}}, \bibinfo {author}
  {\bibfnamefont {F.}~\bibnamefont {Burtscher}}, \bibinfo {author}
  {\bibfnamefont {K.}~\bibnamefont {Morrison}},\ and\ \bibinfo {author}
  {\bibfnamefont {C.}~\bibnamefont {Curto}},\ }in\ \href
  {https://doi.org/10.1007/978-3-030-95519-9_6} {{\selectlanguage
  {English}\emph {\bibinfo {booktitle} {Research in {Computational} {Topology}
  2}}}},\ \bibinfo {editor} {edited by\ \bibinfo {editor} {\bibfnamefont
  {E.}~\bibnamefont {Gasparovic}}, \bibinfo {editor} {\bibfnamefont
  {V.}~\bibnamefont {Robins}},\ and\ \bibinfo {editor} {\bibfnamefont
  {K.}~\bibnamefont {Turner}}}\ (\bibinfo  {publisher} {Springer International
  Publishing},\ \bibinfo {address} {Cham},\ \bibinfo {year} {2022})\ pp.\
  \bibinfo {pages} {129--165}\BibitemShut {NoStop}%
\bibitem [{\citenamefont {Fino}\ and\ \citenamefont
  {Yuste}(2011)}]{fino_dense_2011}%
  \BibitemOpen
  \bibfield  {author} {\bibinfo {author} {\bibfnamefont {E.}~\bibnamefont
  {Fino}}\ and\ \bibinfo {author} {\bibfnamefont {R.}~\bibnamefont {Yuste}},\
  }\href {https://doi.org/10.1016/j.neuron.2011.02.025} {\bibfield  {journal}
  {\bibinfo  {journal} {Neuron}\ }\textbf {\bibinfo {volume} {69}},\ \bibinfo
  {pages} {1188} (\bibinfo {year} {2011})}\BibitemShut {NoStop}%
\bibitem [{\citenamefont {Whittington}\ \emph {et~al.}(2000)\citenamefont
  {Whittington}, \citenamefont {Traub}, \citenamefont {Kopell}, \citenamefont
  {Ermentrout},\ and\ \citenamefont
  {Buhl}}]{whittington_inhibition-based_2000}%
  \BibitemOpen
  \bibfield  {author} {\bibinfo {author} {\bibfnamefont {M.~A.}\ \bibnamefont
  {Whittington}}, \bibinfo {author} {\bibfnamefont {R.~D.}\ \bibnamefont
  {Traub}}, \bibinfo {author} {\bibfnamefont {N.}~\bibnamefont {Kopell}},
  \bibinfo {author} {\bibfnamefont {B.}~\bibnamefont {Ermentrout}},\ and\
  \bibinfo {author} {\bibfnamefont {E.~H.}\ \bibnamefont {Buhl}},\ }\href
  {https://doi.org/10.1016/S0167-8760(00)00173-2} {\bibfield  {journal}
  {\bibinfo  {journal} {International Journal of Psychophysiology}\ }\textbf
  {\bibinfo {volume} {38}},\ \bibinfo {pages} {315} (\bibinfo {year}
  {2000})}\BibitemShut {NoStop}%
\bibitem [{\citenamefont {Stiefel}(2023)}]{stiefel_mean-field_2023}%
  \BibitemOpen
  \bibfield  {author} {\bibinfo {author} {\bibfnamefont {J.}~\bibnamefont
  {Stiefel}},\ }\href {https://doi.org/10.48550/arXiv.2207.02655}
  {{\selectlanguage {English}\bibinfo {title} {Mean-field limits for non-linear
  {Hawkes} processes with inhibition on a {Erdős}-{Rényi}-graph}}} (\bibinfo
  {year} {2023}),\ \bibinfo {note} {arXiv:2207.02655 [math]}\BibitemShut
  {NoStop}%
\bibitem [{\citenamefont {Pfaffelhuber}\ \emph {et~al.}(2022)\citenamefont
  {Pfaffelhuber}, \citenamefont {Rotter},\ and\ \citenamefont
  {Stiefel}}]{pfaffelhuber_mean-field_2022}%
  \BibitemOpen
  \bibfield  {author} {\bibinfo {author} {\bibfnamefont {P.}~\bibnamefont
  {Pfaffelhuber}}, \bibinfo {author} {\bibfnamefont {S.}~\bibnamefont
  {Rotter}},\ and\ \bibinfo {author} {\bibfnamefont {J.}~\bibnamefont
  {Stiefel}},\ }\href {https://doi.org/10.1016/j.spa.2022.07.006} {\bibfield
  {journal} {\bibinfo  {journal} {Stochastic Processes and their Applications}\
  }\textbf {\bibinfo {volume} {153}},\ \bibinfo {pages} {57} (\bibinfo {year}
  {2022})}\BibitemShut {NoStop}%
\bibitem [{\citenamefont {Zhu}(2015)}]{zhu_large_2015}%
  \BibitemOpen
  \bibfield  {author} {\bibinfo {author} {\bibfnamefont {L.}~\bibnamefont
  {Zhu}},\ }\href {https://doi.org/10.1214/14-AAP1003} {\bibfield  {journal}
  {\bibinfo  {journal} {The Annals of Applied Probability}\ }\textbf {\bibinfo
  {volume} {25}},\ \bibinfo {pages} {548} (\bibinfo {year} {2015})},\ \bibinfo
  {note} {publisher: Institute of Mathematical Statistics}\BibitemShut
  {NoStop}%
\bibitem [{\citenamefont {Chevallier}(2017)}]{chevallier_mean-field_2017}%
  \BibitemOpen
  \bibfield  {author} {\bibinfo {author} {\bibfnamefont {J.}~\bibnamefont
  {Chevallier}},\ }\href {https://doi.org/10.1016/j.spa.2017.02.012} {\bibfield
   {journal} {\bibinfo  {journal} {Stochastic Processes and their
  Applications}\ }\textbf {\bibinfo {volume} {127}},\ \bibinfo {pages} {3870}
  (\bibinfo {year} {2017})}\BibitemShut {NoStop}%
\bibitem [{\citenamefont {Delattre}\ \emph {et~al.}(2016)\citenamefont
  {Delattre}, \citenamefont {Fournier},\ and\ \citenamefont
  {Hoffmann}}]{delattre_hawkes_2016}%
  \BibitemOpen
  \bibfield  {author} {\bibinfo {author} {\bibfnamefont {S.}~\bibnamefont
  {Delattre}}, \bibinfo {author} {\bibfnamefont {N.}~\bibnamefont {Fournier}},\
  and\ \bibinfo {author} {\bibfnamefont {M.}~\bibnamefont {Hoffmann}},\ }\href
  {https://doi.org/10.1214/14-AAP1089} {\bibfield  {journal} {\bibinfo
  {journal} {The Annals of Applied Probability}\ }\textbf {\bibinfo {volume}
  {26}},\ \bibinfo {pages} {216} (\bibinfo {year} {2016})},\ \bibinfo {note}
  {publisher: Institute of Mathematical Statistics}\BibitemShut {NoStop}%
\bibitem [{\citenamefont {Heesen}\ and\ \citenamefont
  {Stannat}(2021)}]{heesen_fluctuation_2021}%
  \BibitemOpen
  \bibfield  {author} {\bibinfo {author} {\bibfnamefont {S.}~\bibnamefont
  {Heesen}}\ and\ \bibinfo {author} {\bibfnamefont {W.}~\bibnamefont
  {Stannat}},\ }\href {https://doi.org/10.1016/j.spa.2021.05.007} {\bibfield
  {journal} {\bibinfo  {journal} {Stochastic Processes and their Applications}\
  }\textbf {\bibinfo {volume} {139}},\ \bibinfo {pages} {280} (\bibinfo {year}
  {2021})}\BibitemShut {NoStop}%
\bibitem [{\citenamefont {Ditlevsen}\ and\ \citenamefont
  {Löcherbach}(2017)}]{ditlevsen_multi-class_2017}%
  \BibitemOpen
  \bibfield  {author} {\bibinfo {author} {\bibfnamefont {S.}~\bibnamefont
  {Ditlevsen}}\ and\ \bibinfo {author} {\bibfnamefont {E.}~\bibnamefont
  {Löcherbach}},\ }\href {https://doi.org/10.1016/j.spa.2016.09.013}
  {\bibfield  {journal} {\bibinfo  {journal} {Stochastic Processes and their
  Applications}\ }\textbf {\bibinfo {volume} {127}},\ \bibinfo {pages} {1840}
  (\bibinfo {year} {2017})}\BibitemShut {NoStop}%
\bibitem [{\citenamefont {Miller}\ and\ \citenamefont
  {Fumarola}(2012)}]{miller_mathematical_2012}%
  \BibitemOpen
  \bibfield  {author} {\bibinfo {author} {\bibfnamefont {K.~D.}\ \bibnamefont
  {Miller}}\ and\ \bibinfo {author} {\bibfnamefont {F.}~\bibnamefont
  {Fumarola}},\ }\href {https://doi.org/10.1162/NECO_a_00221} {\bibfield
  {journal} {\bibinfo  {journal} {Neural computation}\ }\textbf {\bibinfo
  {volume} {24}},\ \bibinfo {pages} {25} (\bibinfo {year} {2012})}\BibitemShut
  {NoStop}%
\bibitem [{\citenamefont {Hahnloser}\ \emph {et~al.}(2000)\citenamefont
  {Hahnloser}, \citenamefont {Sarpeshkar}, \citenamefont {Mahowald},
  \citenamefont {Douglas},\ and\ \citenamefont
  {Seung}}]{hahnloser_digital_2000}%
  \BibitemOpen
  \bibfield  {author} {\bibinfo {author} {\bibfnamefont {R.~H.~R.}\
  \bibnamefont {Hahnloser}}, \bibinfo {author} {\bibfnamefont {R.}~\bibnamefont
  {Sarpeshkar}}, \bibinfo {author} {\bibfnamefont {M.~A.}\ \bibnamefont
  {Mahowald}}, \bibinfo {author} {\bibfnamefont {R.~J.}\ \bibnamefont
  {Douglas}},\ and\ \bibinfo {author} {\bibfnamefont {H.~S.}\ \bibnamefont
  {Seung}},\ }\href {https://doi.org/10.1038/35016072} {\bibfield  {journal}
  {\bibinfo  {journal} {Nature}\ }\textbf {\bibinfo {volume} {405}},\ \bibinfo
  {pages} {947} (\bibinfo {year} {2000})},\ \bibinfo {note} {publisher: Nature
  Publishing Group}\BibitemShut {NoStop}%
\bibitem [{\citenamefont {Moon}\ and\ \citenamefont
  {Moser}(1965)}]{moon_cliques_1965}%
  \BibitemOpen
  \bibfield  {author} {\bibinfo {author} {\bibfnamefont {J.~W.}\ \bibnamefont
  {Moon}}\ and\ \bibinfo {author} {\bibfnamefont {L.}~\bibnamefont {Moser}},\
  }\href {https://doi.org/10.1007/BF02760024} {\bibfield  {journal} {\bibinfo
  {journal} {Israel Journal of Mathematics}\ }\textbf {\bibinfo {volume} {3}},\
  \bibinfo {pages} {23} (\bibinfo {year} {1965})}\BibitemShut {NoStop}%
\bibitem [{\citenamefont {Lienkaemper}(2022)}]{lienkaemper_combinatorial_2022}%
  \BibitemOpen
  \bibfield  {author} {\bibinfo {author} {\bibfnamefont {C.}~\bibnamefont
  {Lienkaemper}},\ }\emph {\bibinfo {title} {Combinatorial geometry of neural
  codes, neural data analysis, and neural networks}},\ \href
  {https://etda.libraries.psu.edu/catalog/21234cul434} {Ph.D. thesis},\
  \bibinfo  {school} {Pennsylvania State University} (\bibinfo {year}
  {2022})\BibitemShut {NoStop}%
\bibitem [{\citenamefont {Bel}\ \emph {et~al.}(2021)\citenamefont {Bel},
  \citenamefont {Cobiaga}, \citenamefont {Reartes},\ and\ \citenamefont
  {Rotstein}}]{bel_periodic_2021}%
  \BibitemOpen
  \bibfield  {author} {\bibinfo {author} {\bibfnamefont {A.}~\bibnamefont
  {Bel}}, \bibinfo {author} {\bibfnamefont {R.}~\bibnamefont {Cobiaga}},
  \bibinfo {author} {\bibfnamefont {W.}~\bibnamefont {Reartes}},\ and\ \bibinfo
  {author} {\bibfnamefont {H.~G.}\ \bibnamefont {Rotstein}},\ }\href
  {https://doi.org/10.1137/20M1337831} {\bibfield  {journal} {\bibinfo
  {journal} {SIAM Journal on Applied Dynamical Systems}\ }\textbf {\bibinfo
  {volume} {20}},\ \bibinfo {pages} {1177} (\bibinfo {year}
  {2021})}\BibitemShut {NoStop}%
\bibitem [{\citenamefont {Sadeh}\ and\ \citenamefont
  {Clopath}(2021)}]{sadeh_inhibitory_2021}%
  \BibitemOpen
  \bibfield  {author} {\bibinfo {author} {\bibfnamefont {S.}~\bibnamefont
  {Sadeh}}\ and\ \bibinfo {author} {\bibfnamefont {C.}~\bibnamefont
  {Clopath}},\ }\href {https://doi.org/10.1038/s41583-020-00390-z} {\bibfield
  {journal} {\bibinfo  {journal} {Nature Reviews Neuroscience}\ }\textbf
  {\bibinfo {volume} {22}},\ \bibinfo {pages} {21} (\bibinfo {year} {2021})},\
  \bibinfo {note} {publisher: Nature Publishing Group}\BibitemShut {NoStop}%
\bibitem [{\citenamefont {Tsodyks}\ \emph {et~al.}(1997)\citenamefont
  {Tsodyks}, \citenamefont {Skaggs}, \citenamefont {Sejnowski},\ and\
  \citenamefont {McNaughton}}]{tsodyks_paradoxical_1997}%
  \BibitemOpen
  \bibfield  {author} {\bibinfo {author} {\bibfnamefont {M.~V.}\ \bibnamefont
  {Tsodyks}}, \bibinfo {author} {\bibfnamefont {W.~E.}\ \bibnamefont {Skaggs}},
  \bibinfo {author} {\bibfnamefont {T.~J.}\ \bibnamefont {Sejnowski}},\ and\
  \bibinfo {author} {\bibfnamefont {B.~L.}\ \bibnamefont {McNaughton}},\ }\href
  {https://doi.org/10.1523/JNEUROSCI.17-11-04382.1997} {\bibfield  {journal}
  {\bibinfo  {journal} {The Journal of Neuroscience: The Official Journal of
  the Society for Neuroscience}\ }\textbf {\bibinfo {volume} {17}},\ \bibinfo
  {pages} {4382} (\bibinfo {year} {1997})}\BibitemShut {NoStop}%
\bibitem [{\citenamefont {Ermentrout}\ and\ \citenamefont
  {Terman}(2010)}]{ermentrout_firing_2010}%
  \BibitemOpen
  \bibfield  {author} {\bibinfo {author} {\bibfnamefont {G.~B.}\ \bibnamefont
  {Ermentrout}}\ and\ \bibinfo {author} {\bibfnamefont {D.~H.}\ \bibnamefont
  {Terman}},\ }in\ \href {https://doi.org/10.1007/978-0-387-87708-2_11}
  {{\selectlanguage {English}\emph {\bibinfo {booktitle} {Mathematical
  {Foundations} of {Neuroscience}}}}},\ \bibinfo {editor} {edited by\ \bibinfo
  {editor} {\bibfnamefont {G.~B.}\ \bibnamefont {Ermentrout}}\ and\ \bibinfo
  {editor} {\bibfnamefont {D.~H.}\ \bibnamefont {Terman}}}\ (\bibinfo
  {publisher} {Springer},\ \bibinfo {address} {New York, NY},\ \bibinfo {year}
  {2010})\ pp.\ \bibinfo {pages} {331--367}\BibitemShut {NoStop}%
\bibitem [{\citenamefont {Ermentrout}(1994)}]{ermentrout_reduction_1994}%
  \BibitemOpen
  \bibfield  {author} {\bibinfo {author} {\bibfnamefont {B.}~\bibnamefont
  {Ermentrout}},\ }\href {https://doi.org/10.1162/neco.1994.6.4.679} {\bibfield
   {journal} {\bibinfo  {journal} {Neural Computation}\ }\textbf {\bibinfo
  {volume} {6}},\ \bibinfo {pages} {679} (\bibinfo {year} {1994})}\BibitemShut
  {NoStop}%
\bibitem [{\citenamefont {Rinzel}\ and\ \citenamefont
  {Frankel}(1992)}]{rinzel_activity_1992}%
  \BibitemOpen
  \bibfield  {author} {\bibinfo {author} {\bibfnamefont {J.}~\bibnamefont
  {Rinzel}}\ and\ \bibinfo {author} {\bibfnamefont {P.}~\bibnamefont
  {Frankel}},\ }\href {https://doi.org/10.1162/neco.1992.4.4.534} {\bibfield
  {journal} {\bibinfo  {journal} {Neural Computation}\ }\textbf {\bibinfo
  {volume} {4}},\ \bibinfo {pages} {534} (\bibinfo {year} {1992})}\BibitemShut
  {NoStop}%
\bibitem [{\citenamefont {Shriki}\ \emph {et~al.}(2003)\citenamefont {Shriki},
  \citenamefont {Hansel},\ and\ \citenamefont
  {Sompolinsky}}]{shriki_rate_2003}%
  \BibitemOpen
  \bibfield  {author} {\bibinfo {author} {\bibfnamefont {O.}~\bibnamefont
  {Shriki}}, \bibinfo {author} {\bibfnamefont {D.}~\bibnamefont {Hansel}},\
  and\ \bibinfo {author} {\bibfnamefont {H.}~\bibnamefont {Sompolinsky}},\
  }\href {https://doi.org/10.1162/08997660360675053} {\bibfield  {journal}
  {\bibinfo  {journal} {Neural Computation}\ }\textbf {\bibinfo {volume}
  {15}},\ \bibinfo {pages} {1809} (\bibinfo {year} {2003})}\BibitemShut
  {NoStop}%
\bibitem [{\citenamefont {Pinto}\ and\ \citenamefont
  {Ermentrout}(2001)}]{pinto_spatially_2001}%
  \BibitemOpen
  \bibfield  {author} {\bibinfo {author} {\bibfnamefont {D.~J.}\ \bibnamefont
  {Pinto}}\ and\ \bibinfo {author} {\bibfnamefont {G.~B.}\ \bibnamefont
  {Ermentrout}},\ }\href {https://www.jstor.org/stable/3061903} {\bibfield
  {journal} {\bibinfo  {journal} {SIAM Journal on Applied Mathematics}\
  }\textbf {\bibinfo {volume} {62}},\ \bibinfo {pages} {206} (\bibinfo {year}
  {2001})},\ \bibinfo {note} {publisher: Society for Industrial and Applied
  Mathematics}\BibitemShut {NoStop}%
\bibitem [{\citenamefont {Aviel}\ and\ \citenamefont
  {Gerstner}(2006)}]{aviel_spiking_2006}%
  \BibitemOpen
  \bibfield  {author} {\bibinfo {author} {\bibfnamefont {Y.}~\bibnamefont
  {Aviel}}\ and\ \bibinfo {author} {\bibfnamefont {W.}~\bibnamefont
  {Gerstner}},\ }\href {https://doi.org/10.1103/PhysRevE.73.051908} {\bibfield
  {journal} {\bibinfo  {journal} {Physical Review E}\ }\textbf {\bibinfo
  {volume} {73}},\ \bibinfo {pages} {051908} (\bibinfo {year} {2006})},\
  \bibinfo {note} {publisher: American Physical Society}\BibitemShut {NoStop}%
\bibitem [{\citenamefont {Ostojic}\ and\ \citenamefont
  {Fusi}(2024)}]{ostojic_computational_2024}%
  \BibitemOpen
  \bibfield  {author} {\bibinfo {author} {\bibfnamefont {S.}~\bibnamefont
  {Ostojic}}\ and\ \bibinfo {author} {\bibfnamefont {S.}~\bibnamefont {Fusi}},\
  }\href {https://doi.org/10.1016/j.tics.2024.03.003} {\bibfield  {journal}
  {\bibinfo  {journal} {Trends in Cognitive Sciences}\ }\textbf {\bibinfo
  {volume} {28}},\ \bibinfo {pages} {677} (\bibinfo {year} {2024})},\ \bibinfo
  {note} {publisher: Elsevier}\BibitemShut {NoStop}%
\bibitem [{\citenamefont
  {Londono~Alvarez}(2024)}]{londono_alvarez_juliana_attractor-based_2024}%
  \BibitemOpen
  \bibfield  {author} {\bibinfo {author} {\bibfnamefont {J.}~\bibnamefont
  {Londono~Alvarez}},\ }{\selectlanguage {English}\emph {\bibinfo {title}
  {Attractor-based models for sequences and pattern generation in neural
  circuits}}},\ \href {https://etda.libraries.psu.edu/catalog/24154jbl5958}
  {Ph.D. thesis},\ \bibinfo  {school} {Pennsylvania State University}, \bibinfo
  {address} {University Park, PA} (\bibinfo {year} {2024})\BibitemShut
  {NoStop}%
\bibitem [{\citenamefont {Mazzucato}\ \emph {et~al.}(2016)\citenamefont
  {Mazzucato}, \citenamefont {Fontanini},\ and\ \citenamefont
  {La~Camera}}]{mazzucato_stimuli_2016}%
  \BibitemOpen
  \bibfield  {author} {\bibinfo {author} {\bibfnamefont {L.}~\bibnamefont
  {Mazzucato}}, \bibinfo {author} {\bibfnamefont {A.}~\bibnamefont
  {Fontanini}},\ and\ \bibinfo {author} {\bibfnamefont {G.}~\bibnamefont
  {La~Camera}},\ }\bibfield  {journal} {\bibinfo  {journal} {Frontiers in
  Systems Neuroscience}\ }\textbf {\bibinfo {volume} {10}},\ \href
  {https://doi.org/10.3389/fnsys.2016.00011} {10.3389/fnsys.2016.00011}
  (\bibinfo {year} {2016}),\ \bibinfo {note} {publisher: Frontiers}\BibitemShut
  {NoStop}%
\bibitem [{\citenamefont {Schwalger}\ and\ \citenamefont
  {Chizhov}(2019)}]{schwalger_mind_2019}%
  \BibitemOpen
  \bibfield  {author} {\bibinfo {author} {\bibfnamefont {T.}~\bibnamefont
  {Schwalger}}\ and\ \bibinfo {author} {\bibfnamefont {A.~V.}\ \bibnamefont
  {Chizhov}},\ }\href {https://doi.org/10.1016/j.conb.2019.08.003} {\bibfield
  {journal} {\bibinfo  {journal} {Current Opinion in Neurobiology}\ }\bibinfo
  {series} {Computational {Neuroscience}},\ \textbf {\bibinfo {volume} {58}},\
  \bibinfo {pages} {155} (\bibinfo {year} {2019})}\BibitemShut {NoStop}%
\bibitem [{\citenamefont {Freidlin}\ and\ \citenamefont
  {Wentzell}(2012)}]{freidlin_random_2012}%
  \BibitemOpen
  \bibfield  {author} {\bibinfo {author} {\bibfnamefont {M.~I.}\ \bibnamefont
  {Freidlin}}\ and\ \bibinfo {author} {\bibfnamefont {A.~D.}\ \bibnamefont
  {Wentzell}},\ }in\ \href {https://doi.org/10.1007/978-3-642-25847-3_8}
  {{\selectlanguage {English}\emph {\bibinfo {booktitle} {Random
  {Perturbations} of {Dynamical} {Systems}}}}},\ \bibinfo {editor} {edited by\
  \bibinfo {editor} {\bibfnamefont {M.~I.}\ \bibnamefont {Freidlin}}\ and\
  \bibinfo {editor} {\bibfnamefont {A.~D.}\ \bibnamefont {Wentzell}}}\
  (\bibinfo  {publisher} {Springer},\ \bibinfo {address} {Berlin, Heidelberg},\
  \bibinfo {year} {2012})\ pp.\ \bibinfo {pages} {258--354}\BibitemShut
  {NoStop}%
\bibitem [{\citenamefont {Bressloff}\ and\ \citenamefont
  {Newby}(2014)}]{bressloff_path_2014}%
  \BibitemOpen
  \bibfield  {author} {\bibinfo {author} {\bibfnamefont {P.~C.}\ \bibnamefont
  {Bressloff}}\ and\ \bibinfo {author} {\bibfnamefont {J.~M.}\ \bibnamefont
  {Newby}},\ }\href {https://doi.org/10.1103/PhysRevE.89.042701} {\bibfield
  {journal} {\bibinfo  {journal} {Physical Review E}\ }\textbf {\bibinfo
  {volume} {89}},\ \bibinfo {pages} {042701} (\bibinfo {year}
  {2014})}\BibitemShut {NoStop}%
\bibitem [{\citenamefont {Fleurantin}\ \emph {et~al.}(2023)\citenamefont
  {Fleurantin}, \citenamefont {Slyman}, \citenamefont {Barker},\ and\
  \citenamefont {Jones}}]{fleurantin_dynamical_2023}%
  \BibitemOpen
  \bibfield  {author} {\bibinfo {author} {\bibfnamefont {E.}~\bibnamefont
  {Fleurantin}}, \bibinfo {author} {\bibfnamefont {K.}~\bibnamefont {Slyman}},
  \bibinfo {author} {\bibfnamefont {B.}~\bibnamefont {Barker}},\ and\ \bibinfo
  {author} {\bibfnamefont {C.~K. R.~T.}\ \bibnamefont {Jones}},\ }\href
  {https://doi.org/10.1016/j.physd.2023.133860} {\bibfield  {journal} {\bibinfo
   {journal} {Physica D: Nonlinear Phenomena}\ }\textbf {\bibinfo {volume}
  {454}},\ \bibinfo {pages} {133860} (\bibinfo {year} {2023})}\BibitemShut
  {NoStop}%
\bibitem [{\citenamefont {Barral}\ and\ \citenamefont
  {D~Reyes}(2016)}]{barral_synaptic_2016}%
  \BibitemOpen
  \bibfield  {author} {\bibinfo {author} {\bibfnamefont {J.}~\bibnamefont
  {Barral}}\ and\ \bibinfo {author} {\bibfnamefont {A.}~\bibnamefont
  {D~Reyes}},\ }\href {https://doi.org/10.1038/nn.4415} {\bibfield  {journal}
  {\bibinfo  {journal} {Nature Neuroscience}\ }\textbf {\bibinfo {volume}
  {19}},\ \bibinfo {pages} {1690} (\bibinfo {year} {2016})},\ \bibinfo {note}
  {publisher: Nature Publishing Group}\BibitemShut {NoStop}%
\bibitem [{\citenamefont {Sompolinsky}\ \emph {et~al.}(1988)\citenamefont
  {Sompolinsky}, \citenamefont {Crisanti},\ and\ \citenamefont
  {Sommers}}]{sompolinsky_chaos_1988}%
  \BibitemOpen
  \bibfield  {author} {\bibinfo {author} {\bibfnamefont {H.}~\bibnamefont
  {Sompolinsky}}, \bibinfo {author} {\bibfnamefont {A.}~\bibnamefont
  {Crisanti}},\ and\ \bibinfo {author} {\bibfnamefont {H.~J.}\ \bibnamefont
  {Sommers}},\ }\href {https://doi.org/10.1103/PhysRevLett.61.259} {\bibfield
  {journal} {\bibinfo  {journal} {Physical Review Letters}\ }\textbf {\bibinfo
  {volume} {61}},\ \bibinfo {pages} {259} (\bibinfo {year} {1988})},\ \bibinfo
  {note} {publisher: American Physical Society}\BibitemShut {NoStop}%
\bibitem [{\citenamefont {Zeraati}\ \emph {et~al.}(2024)\citenamefont
  {Zeraati}, \citenamefont {Buendía}, \citenamefont {Engel},\ and\
  \citenamefont {Levina}}]{zeraati_topology-dependent_2024}%
  \BibitemOpen
  \bibfield  {author} {\bibinfo {author} {\bibfnamefont {R.}~\bibnamefont
  {Zeraati}}, \bibinfo {author} {\bibfnamefont {V.}~\bibnamefont {Buendía}},
  \bibinfo {author} {\bibfnamefont {T.~A.}\ \bibnamefont {Engel}},\ and\
  \bibinfo {author} {\bibfnamefont {A.}~\bibnamefont {Levina}},\ }\href
  {https://doi.org/10.1103/PhysRevResearch.6.023131} {\bibfield  {journal}
  {\bibinfo  {journal} {Physical Review Research}\ }\textbf {\bibinfo {volume}
  {6}},\ \bibinfo {pages} {023131} (\bibinfo {year} {2024})}\BibitemShut
  {NoStop}%
\bibitem [{\citenamefont {Mastrogiuseppe}\ and\ \citenamefont
  {Ostojic}(2018)}]{mastrogiuseppe_linking_2018}%
  \BibitemOpen
  \bibfield  {author} {\bibinfo {author} {\bibfnamefont {F.}~\bibnamefont
  {Mastrogiuseppe}}\ and\ \bibinfo {author} {\bibfnamefont {S.}~\bibnamefont
  {Ostojic}},\ }\href {https://doi.org/10.1016/j.neuron.2018.07.003} {\bibfield
   {journal} {\bibinfo  {journal} {Neuron}\ }\textbf {\bibinfo {volume} {99}},\
  \bibinfo {pages} {609} (\bibinfo {year} {2018})},\ \bibinfo {note}
  {publisher: Elsevier}\BibitemShut {NoStop}%
\bibitem [{\citenamefont {Dahmen}\ \emph {et~al.}(2023)\citenamefont {Dahmen},
  \citenamefont {Recanatesi}, \citenamefont {Jia}, \citenamefont {Ocker},
  \citenamefont {Campagnola}, \citenamefont {Seeman}, \citenamefont {Jarsky},
  \citenamefont {Helias},\ and\ \citenamefont
  {Shea-Brown}}]{dahmen_strong_2023}%
  \BibitemOpen
  \bibfield  {author} {\bibinfo {author} {\bibfnamefont {D.}~\bibnamefont
  {Dahmen}}, \bibinfo {author} {\bibfnamefont {S.}~\bibnamefont {Recanatesi}},
  \bibinfo {author} {\bibfnamefont {X.}~\bibnamefont {Jia}}, \bibinfo {author}
  {\bibfnamefont {G.~K.}\ \bibnamefont {Ocker}}, \bibinfo {author}
  {\bibfnamefont {L.}~\bibnamefont {Campagnola}}, \bibinfo {author}
  {\bibfnamefont {S.}~\bibnamefont {Seeman}}, \bibinfo {author} {\bibfnamefont
  {T.}~\bibnamefont {Jarsky}}, \bibinfo {author} {\bibfnamefont
  {M.}~\bibnamefont {Helias}},\ and\ \bibinfo {author} {\bibfnamefont
  {E.}~\bibnamefont {Shea-Brown}},\ }\href
  {https://doi.org/10.1101/2020.11.02.365072} {{\selectlanguage
  {English}\bibinfo {title} {Strong and localized recurrence controls
  dimensionality of neural activity across brain areas}}} (\bibinfo {year}
  {2023}),\ \bibinfo {note} {pages: 2020.11.02.365072 Section: New
  Results}\BibitemShut {NoStop}%
\bibitem [{\citenamefont {Hu}\ \emph {et~al.}(2013)\citenamefont {Hu},
  \citenamefont {Trousdale}, \citenamefont {Josić},\ and\ \citenamefont
  {Shea-Brown}}]{hu_motif_2013}%
  \BibitemOpen
  \bibfield  {author} {\bibinfo {author} {\bibfnamefont {Y.}~\bibnamefont
  {Hu}}, \bibinfo {author} {\bibfnamefont {J.}~\bibnamefont {Trousdale}},
  \bibinfo {author} {\bibfnamefont {K.}~\bibnamefont {Josić}},\ and\ \bibinfo
  {author} {\bibfnamefont {E.}~\bibnamefont {Shea-Brown}},\ }\href
  {https://doi.org/10.1088/1742-5468/2013/03/P03012} {\bibfield  {journal}
  {\bibinfo  {journal} {Journal of Statistical Mechanics: Theory and
  Experiment}\ }\textbf {\bibinfo {volume} {2013}},\ \bibinfo {pages} {P03012}
  (\bibinfo {year} {2013})},\ \bibinfo {note} {publisher: IOP Publishing and
  SISSA}\BibitemShut {NoStop}%
\bibitem [{\citenamefont {Hu}\ and\ \citenamefont
  {Sompolinsky}(2022)}]{hu_spectrum_2022}%
  \BibitemOpen
  \bibfield  {author} {\bibinfo {author} {\bibfnamefont {Y.}~\bibnamefont
  {Hu}}\ and\ \bibinfo {author} {\bibfnamefont {H.}~\bibnamefont
  {Sompolinsky}},\ }\href {https://doi.org/10.1371/journal.pcbi.1010327}
  {\bibfield  {journal} {\bibinfo  {journal} {PLOS Computational Biology}\
  }\textbf {\bibinfo {volume} {18}},\ \bibinfo {pages} {e1010327} (\bibinfo
  {year} {2022})},\ \bibinfo {note} {publisher: Public Library of
  Science}\BibitemShut {NoStop}%
\bibitem [{\citenamefont {Aljadeff}\ \emph {et~al.}(2015)\citenamefont
  {Aljadeff}, \citenamefont {Stern},\ and\ \citenamefont
  {Sharpee}}]{aljadeff_transition_2015}%
  \BibitemOpen
  \bibfield  {author} {\bibinfo {author} {\bibfnamefont {J.}~\bibnamefont
  {Aljadeff}}, \bibinfo {author} {\bibfnamefont {M.}~\bibnamefont {Stern}},\
  and\ \bibinfo {author} {\bibfnamefont {T.}~\bibnamefont {Sharpee}},\ }\href
  {https://doi.org/10.1103/PhysRevLett.114.088101} {\bibfield  {journal}
  {\bibinfo  {journal} {Physical Review Letters}\ }\textbf {\bibinfo {volume}
  {114}},\ \bibinfo {pages} {088101} (\bibinfo {year} {2015})},\ \bibinfo
  {note} {publisher: American Physical Society}\BibitemShut {NoStop}%
\bibitem [{\citenamefont {Stern}\ \emph {et~al.}(2014)\citenamefont {Stern},
  \citenamefont {Sompolinsky},\ and\ \citenamefont
  {Abbott}}]{stern_dynamics_2014}%
  \BibitemOpen
  \bibfield  {author} {\bibinfo {author} {\bibfnamefont {M.}~\bibnamefont
  {Stern}}, \bibinfo {author} {\bibfnamefont {H.}~\bibnamefont {Sompolinsky}},\
  and\ \bibinfo {author} {\bibfnamefont {L.~F.}\ \bibnamefont {Abbott}},\
  }\href {https://doi.org/10.1103/PhysRevE.90.062710} {\bibfield  {journal}
  {\bibinfo  {journal} {Physical Review E}\ }\textbf {\bibinfo {volume} {90}},\
  \bibinfo {pages} {062710} (\bibinfo {year} {2014})},\ \bibinfo {note}
  {publisher: American Physical Society}\BibitemShut {NoStop}%
\bibitem [{\citenamefont {Clark}\ and\ \citenamefont
  {Beiran}(2025)}]{clark_structure_2025}%
  \BibitemOpen
  \bibfield  {author} {\bibinfo {author} {\bibfnamefont {D.~G.}\ \bibnamefont
  {Clark}}\ and\ \bibinfo {author} {\bibfnamefont {M.}~\bibnamefont {Beiran}},\
  }\href {https://doi.org/10.1073/pnas.2404039122} {\bibfield  {journal}
  {\bibinfo  {journal} {Proceedings of the National Academy of Sciences}\
  }\textbf {\bibinfo {volume} {122}},\ \bibinfo {pages} {e2404039122} (\bibinfo
  {year} {2025})},\ \bibinfo {note} {publisher: Proceedings of the National
  Academy of Sciences}\BibitemShut {NoStop}%
\bibitem [{\citenamefont {Wardak}\ and\ \citenamefont
  {Gong}(2022)}]{wardak_extended_2022}%
  \BibitemOpen
  \bibfield  {author} {\bibinfo {author} {\bibfnamefont {A.}~\bibnamefont
  {Wardak}}\ and\ \bibinfo {author} {\bibfnamefont {P.}~\bibnamefont {Gong}},\
  }\href {https://doi.org/10.1103/PhysRevLett.129.048103} {\bibfield  {journal}
  {\bibinfo  {journal} {Physical Review Letters}\ }\textbf {\bibinfo {volume}
  {129}},\ \bibinfo {pages} {048103} (\bibinfo {year} {2022})},\ \bibinfo
  {note} {publisher: American Physical Society}\BibitemShut {NoStop}%
\bibitem [{\citenamefont {Helias}\ and\ \citenamefont
  {Dahmen}(2020)}]{helias_statistical_2020}%
  \BibitemOpen
  \bibfield  {author} {\bibinfo {author} {\bibfnamefont {M.}~\bibnamefont
  {Helias}}\ and\ \bibinfo {author} {\bibfnamefont {D.}~\bibnamefont
  {Dahmen}},\ }\href {https://doi.org/10.1007/978-3-030-46444-8}
  {{\selectlanguage {English}\emph {\bibinfo {title} {Statistical {Field}
  {Theory} for {Neural} {Networks}}}}},\ \bibinfo {series} {Lecture {Notes} in
  {Physics}}, Vol.\ \bibinfo {volume} {970}\ (\bibinfo  {publisher} {Springer
  International Publishing},\ \bibinfo {address} {Cham},\ \bibinfo {year}
  {2020})\BibitemShut {NoStop}%
\bibitem [{\citenamefont {Ocker}(2023)}]{ocker_republished_2023}%
  \BibitemOpen
  \bibfield  {author} {\bibinfo {author} {\bibfnamefont {G.~K.}\ \bibnamefont
  {Ocker}},\ }\href {https://doi.org/10.1103/PhysRevX.13.041047} {\bibfield
  {journal} {\bibinfo  {journal} {Physical Review X}\ }\textbf {\bibinfo
  {volume} {13}},\ \bibinfo {pages} {041047} (\bibinfo {year} {2023})},\
  \bibinfo {note} {publisher: American Physical Society}\BibitemShut {NoStop}%
\bibitem [{\citenamefont {Beer}(2006)}]{beer_parameter_2006}%
  \BibitemOpen
  \bibfield  {author} {\bibinfo {author} {\bibfnamefont {R.~D.}\ \bibnamefont
  {Beer}},\ }\href {https://doi.org/10.1162/neco.2006.18.12.3009} {\bibfield
  {journal} {\bibinfo  {journal} {Neural Computation}\ }\textbf {\bibinfo
  {volume} {18}},\ \bibinfo {pages} {3009} (\bibinfo {year}
  {2006})}\BibitemShut {NoStop}%
\bibitem [{\citenamefont {Veltz}(2020)}]{veltz_bifurcationkitjl_2020}%
  \BibitemOpen
  \bibfield  {author} {\bibinfo {author} {\bibfnamefont {R.}~\bibnamefont
  {Veltz}},\ }\href {https://hal.archives-ouvertes.fr/hal-02902346} {\bibinfo
  {title} {{BifurcationKit}.jl}} (\bibinfo {year} {2020})\BibitemShut {NoStop}%
\end{thebibliography}%

\onecolumngrid
\appendix
\section{Methods}

\subsection{Clustered spiking networks: microscopic model and mean-field limit}
\label{sec:model_details}
Our model consists of $n$ excitatory clusters with $N_E$ neurons each, reciprocally connected to a common inhibitory population of $N_I$ neurons. The network has additional structure given by a directed graph $G$ on $n$ nodes which describes the pattern of connectivity between the excitatory clusters. Neuron-to-neuron connectivity is described by a block-wise Erd\H{o}s-R\'enyi adjacency matrix A with edge probabilities depending on the neuron type as well as the graph structure $G$:

\begin{align*}
p_{\alpha \beta} = \begin{cases} 
p^\circlearrowleft_{EE} &\mbox { if } \beta = \alpha, \alpha, \beta \mbox{ excitatory}\\
p_{EE}^{\to}  &\mbox { if } \beta \to \alpha \in G,  \alpha, \beta  \mbox{ excitatory}\\
p_{EE}^{\not\to}  &\mbox { if } \beta \not\to \alpha \in G,  \alpha, \beta  \mbox{ excitatory}\\
p_{EI} &\mbox { if } \alpha  \mbox{ excitatory}, \beta \mbox{ inhibitory }\\
p_{IE} &\mbox { if } \alpha \mbox{ inhibitory } , \beta \mbox{ excitatory}\\
p_{II} &\mbox { if } \alpha, \beta \mbox{ inhibitory. }
\end{cases}
\end{align*}
Neuron-to-neuron weights are then given by 

\begin{align*}
    J_{ij} = \frac{A_{ij}\bar J_{\alpha(i)\alpha(j)}}{p_{\alpha(i)\alpha(j)} N_{\alpha(j)}},
\end{align*}
where

\begin{align}
\bar J_{\alpha \beta} = \begin{cases} 
\bar J^\circlearrowleft_{EE} &\mbox { if } \beta = \alpha, \alpha, \beta \mbox{ excitatory}\\
\bar J_{EE}^{\to}  &\mbox { if } \beta \to \alpha \in G,  \alpha, \beta  \mbox{ excitatory}\\
\bar J_{EE}^{\not\to}  &\mbox { if } \beta \not\to \alpha \in G,  \alpha, \beta  \mbox{ excitatory}\\
\bar J_{EI} &\mbox { if } \alpha  \mbox{ excitatory}, \beta \mbox{ inhibitory }\\
\bar J_{IE} &\mbox { if } \alpha \mbox{ inhibitory } , \beta \mbox{ excitatory}\\
\bar J_{II} &\mbox { if } \alpha, \beta \mbox{ inhibitory. }
\end{cases}
\label{app:eqn:weights}
\end{align}
 
 By construction, $\abk{\sum_{j\in \alpha} J_{ij}} = \bar J_{\alpha(i)\alpha(j)}.$ Excitatory weights $\bar J^\circlearrowleft_{EE},\bar J_{EE}^{\to}, \bar J_{EE}^{\not\to}$ and $\bar J_{IE}$ are positive, while inhibitory weights $\bar J_{EI}$ and $\bar J_{II}$ are negative. Connectivity within populations is stronger than connectivity between populations, and connectivity between populations is stronger when there is a directed edge than when not: $\bar J_{EE}^\circlearrowleft > \bar J_{EE}^\to > \bar J_{EE}^{\not\to}$.   

Inhibition in the model is strong enough to cancel out the within-cluster excitation, with 

\begin{align*}
    \bar J_{EE}^{\circlearrowleft} = -\frac{\bar J_{EI}\bar J_{IE}}{1-\bar J_{II}}.
\end{align*}
The network's dynamics are given by the nonlinear Hawkes process 
\begin{align}
    \tau_{i} &\dot v_{\alpha(i)} = - v_i + \sum_{j =1}^N J_{ij} \dot s_j + b_{\alpha(i)}, \\
    \dot s_i &\sim \mathcal{PP} \left(\lfloor v_i\rfloor_+\right).
    \label{eqn:meth_eqn_spiking}
\end{align}
Excitatory and inhibitory neurons have distinct time constants $\tau_E$ and $\tau_I$. 

In Fig. \ref{fig:diverse_cluster_dynamics}, we show raster plots for three clustered spiking networks. The adjacency matrices for the graphs $G$ determining connectivity between the excitatory clusters are \begin{align*}
   \bar A_1 =  \begin{pmatrix} 
    0 & 0& 0 & 0 &0 & 0\\
    0 & 0& 0 & 0 &0 & 0\\
    0 & 0& 0 & 0 &0 & 0\\
    0 & 0& 0 & 0 &0 & 0\\
    0 & 0& 0 & 0 &0 & 0\\
    0 & 0& 0 & 0 &0 & 0
    \end{pmatrix} && 
      \bar A_2 =\begin{pmatrix} 
    0 & 1& 0 & 1 &0 & 0\\
    1 & 0& 1 & 1 &0 & 0\\
    0 & 1& 0 & 1 &0 & 0\\
    1 & 1 & 1 & 0 &1 & 0\\
    0 & 0& 0 & 1 &0 & 0\\
    0 & 0& 0 & 0 &0 & 0\end{pmatrix} && 
      \bar A_3 =\begin{pmatrix} 
    0 & 0& 1 & 1 \\
    1 & 0& 0 & 0 \\
    0 & 1& 0 & 0 \\
    0 & 0 & 0 & 1\\
  \end{pmatrix} & 
\end{align*} 
for panels A/B, C/D, and E/F respectively. The simulation parameters are $N_E = N_I = 1000$, $\tau_E = 40  \unit{ms}, \tau_I = 20 \unit{ms}$, $p_{EE}= .2, p_{EI}= .8, p_{IE}= .8,  p_{II}= .8, p_{EE}^\circlearrowleft =  .8$ , $J_{EE}^\circlearrowleft = 1.5$, $J_{EE}^\to = 0.75$, $J_{EE}^{\not\to} = 0$, $J_{EI} = -3.0$, $J_{II} = -3.0$, $J_{IE} = 2.0, b_E = 0.1$, $b_I = 0$.  Note that while $b_E, b_I$ have units of voltage and the synaptic weights $\bar J_{\alpha\beta}$ have units of voltage/(spikes/ms), we absorb the voltage units into the function taking voltage to firing rate, thus we leave these quantities non-dimensional.  At the times marked with lightning bolts ($t = 500, 1000, 1500, 2000, 2500$ in panel B, $t = 1000, 2000, 3000$ in panel D, and $t = 1000, 2000$ in panel F), we applied a transient stimulus to the population or populations indicated by the color of the lightning bolt. In panel B, the stimulus was of duration 100 ms and strength $0.25$. In panel D, the stimulus was of duration $50 \unit{ms}$ and strength 0.25. In panel F,  the stimulus was of duration $40 \unit{ms}$ and strength 0.2. These different stimulus strengths and durations were chosen due to the differing depths of the basins of each network's attractors. 

We next present a heuristic derivation of the mean-field limit for this model in terms of the voltage dynamics.
To describe the evolution of the population-averaged voltages, $v_\alpha :=\abk{v_i}_{i\in \alpha}$\, we take the expected value on both sides of \ref{eqn:meth_eqn_spiking}, obtaining 
\begin{align}
        \tau_{i} &v_{\alpha}= - v_{\alpha}+ \abk{\sum_{j =1}^N J_{ij} \dot s_j}_{i\in \alpha} + b_{\alpha} 
\end{align}
This expectation is being taken over the random connectivity matrix, the stochastic spike emission, and the identity of neuron $i$ within cluster $\alpha$. Now, we make the assumption that in the large-N limit, the connectivity and the spike emission become independent, and note that $\abk{\dot s_j}_{j\in \beta} = \abk {\lfloor v_j\rfloor_+}_{j\in \beta}$. This gives us the approximation
\begin{align}
        \tau_{i} &v_{\alpha}= - v_{\alpha}+ \sum_{\beta = 1}^{n+1} \bar J_{\alpha \beta} \abk{\lfloor v_j\rfloor_+}_{j\in \beta} + b_{\alpha}. 
\end{align}
Finally, we make the assumption that in the large-$N$ limit, the value of $v_j$ concentrates around its expectation. This allows us to switch the order of the expectation and the threshold nonlinearity, obtaining the desired mean- field equation
\begin{align}
    \tau_\alpha \dot v_{\alpha}= -v_{\alpha}  + \sum_{j=1}^N \bar J_{\alpha\beta} \lfloor v_\beta \rfloor_+ + b_\alpha. \label{eqn:v_mean_field}
\end{align}
These mean-field dynamics can also be obtain by averaging the generating functional for Eq.~\ref{eqn:meth_eqn_spiking} over realizations of the connectivity in the large $N$ limit \cite{helias_statistical_2020}. In that case, the neurons' net synaptic inputs concentrate around the mean and the generating functional factorizes across neurons ~\cite{ocker_republished_2023}.
Corresponding mean-field limits for similar nonlinear Hawkes models are proven in \cite{stiefel_mean-field_2023, pfaffelhuber_mean-field_2022, zhu_large_2015, chevallier_mean-field_2017, delattre_hawkes_2016, heesen_fluctuation_2021, ditlevsen_multi-class_2017}.

\subsection{Equivalence between two formulations}
\label{sec:trans}
Miller and Fumarola \cite{miller_mathematical_2012} establish the equivalence of two commonly used models, those of the form $\tau \dot v = -v + J f(v) + b$ and $\tau \dot x = -x + f(J x+ c).$ In the restricted case when $J$ is invertible and $ b$ is constant, this equivalence via the map $b = c$, $v = Jx + c$ was established earlier in \cite{beer_parameter_2006}.  Here, we focus on invertible $J$ (but not necessarily constant $b$) because this is what occurs generically.

In this case, the map is $v = Jx + c$, $x = J\inv (x - c)$, as in \cite{beer_parameter_2006}. However, the input transforms as $\tau \dot c = -c + b$. Notice that this means that the input to the $x$-model is a low-pass filter of the input in the equivalent $v$- model, with $c(t) = \frac 1\tau \int_0^t e^{(t'-t)/\tau} b(t')dt'$. 

In Fig. \ref{fig:mean_field}B, we plot population-mean firing rates, computed by averaging $\lfloor v_i \rfloor_+$ over all neurons in each cluster, against predictions made by simulating the EI-TLN model of Eq. \ref{eqn:ei_tln}. Because we used time varying input, the input to the EI-TLN is a low-pass-filtered version of the mean  input to the spiking model. We then plot  $\lfloor v\rfloor_+ = \lfloor J x + c\rfloor_+$ against the population-averaged rates. 

Hereafter, we only consider the case where the external input $b$ to the $v$-model is constant, so the input to the $x$-model is $c = b$. Thus, we use $b$ to denote external inputs regardless of formulation hereafter. 

This transformation of the input explains why Eq. \ref{eqn:meth_eqn_spiking} is not dynamically equivalent to the $x$-model of the form 
\begin{align}
    \tau_{\alpha (i)} \dot x_i = -x_i+ \left\lfloor \sum_{j=1}^N J_{ij} \dot s_j + b_{\alpha (i)} \right\rfloor_+ . \label{eqn:x_spiking}
\end{align}
 We can re-write Eq. \ref{eqn:v_model_spk} as 
 \begin{align*}
         \tau_{i} &\dot v_{\alpha(i)} = - v_i + \sum_{j =1}^N J_{ij}\lfloor v_j\rfloor_++ \left(b_{\alpha(i)}  + \sum_{i=1}^N J_{ij}(\dot s_j - \lfloor v_j\rfloor_+)\right),
 \end{align*}
 i.e., separating the stochastic spiking into an external input $b_i(t) = b_{\alpha(i)}  + \sum_{i=1}^N J_{ij}(\dot s_j - \lfloor v_j\rfloor_+)$. Then the equivalent input to the $x$- model is the low-pass filtered input $c_i(t)=  \frac 1\tau \int_0^t e^{(t'-t)/\tau} b_i(t')dt'$. Indeed, the $x$-model \ref{eqn:x_spiking} does not have the EI-TLN as a mean-field theory because variance of the input to the threshold nonlinearity, $\sum_{j=1}^N J_{ij} \dot s_j + b_{\alpha (i)}$, has infinite variance which does not converge to any finite value even when $N\to \infty$. This means that it is not possible to exchange the order of the threshold nonlinearity with the expectation, as we did in the derivation of \ref{eqn:v_mean_field}.  A spiking model which has the $x$-model form would need to incorporate an additional filter: 
 \begin{align}
    \tau_{\alpha (i)} \dot x_i = -x_i+ \left \lfloor \sum_{j=1}^N J_{ij} (g *\dot s_j) + b_{\alpha (i)} \right \rfloor_+ ,
\end{align}
where the term $ (g *\dot s_j)$ denotes a convolution, and $g$ is some synaptic filter, i.e. $g(t) = \frac{1}{\tau_s}e^{t/\tau_s}$. We could then obtain the EI-TLN model as a mean-field theory of this model, in the additional limit where the filter width $\tau_s$ goes to zero.

\subsection{The EI-TLN model and the CTLN model}
\label{sec:ei_tln_to_ctln}
In the previous section, we showed that in the limit of large population size, the expected firing rates evolve according to deterministic, $n+1$ dimensional dynamics. We now study this system, and relate it to the CTLN model. Letting $x_\alpha = \abk{x_i}_{i\in \alpha}$ we have 

\begin{align}
\tau_{\alpha}\dot x_{\alpha} = -x_\alpha + \left \lfloor \sum_{\beta \in [n+1]} \bar J_{\alpha\beta}x_{\beta} + b_{\beta}\right \rfloor _+  , 
\label{app:eqn:E_I_ctln}
\end{align}
where weights $\bar J_{\alpha\beta}$ is defined as in Equation \ref{app:eqn:weights}, the $n$ excitatory populations are indexed by $\alpha \in [n]$, and the inhibitory population is indexed as $\alpha = n+1$. We term this model the Excitatory-Inhibitory Threshold Linear Network (EI-TLN) because its dynamics are determined by the directed graph $G$.  

Next, we show that in the limit of fast inhibition, $\tau_I/\tau_E\to 0$, for a certain choice of weights the dynamics on the excitatory units in the EI-TLN model converge to the original CTLN model, introduced in \cite{curto_pattern_2016, curto_fixed_2019, morrison_diversity_2024}. This will make it possible to apply the theory developed relating fixed points and dynamics of CTLNs to properties of graph structure to study the dynamics of clustered spiking networks. 

The CTLN model describes a network with specific excitatory connectivity against a background of global inhibition using a pattern of strong and weak inhibition.  The weights $W_{\alpha\beta}$ in a CTLN are determined by a directed graph $G$, and given by 
\begin{align*}
W_{\alpha\beta} = \begin{cases} 
0 &\mbox { if } \alpha =\beta\\
-1 + \varepsilon &\mbox { if } \beta \to \alpha \mbox { in } G\\
-1 + \delta  &\mbox { if } \beta \not\to \alpha\mbox { in } G, \\
\end{cases}
\end{align*}
where $\varepsilon, \delta$ satisfy $0 \leq \epsilon < \frac{\delta}{\delta+1}$. All neurons in the CTLN receive a uniform input $\theta > 0.$

 In the limit of fast inhibition, $\tau_I/\tau_E \approx 0$, we can approximate the inhibitory neuron's rate as being at its fixed point value, 
\begin{align*}
x_{n+1} = [\bar Jx + \theta_I]_+ = \left \lfloor \bar J_{IE}\sum_{\alpha = 1}^n x_\alpha + \bar J_{II}x_{n+1} + b_{I}\right \rfloor _+
\end{align*}
We will drop the the threshold, assuming $x_I > 0$, and solve for $x_I$, obtaining 
\begin{align*}
    x_{n+1} = \frac{\bar J_{IE} \sum_{\alpha = 1}^n x_\alpha + b_I}{1 - \bar J_{II}}
\end{align*}
When $1 - \bar J_{II} > 0$, $b_I \geq 0$, this yields a positive solution for $x_{n+1}$, thus at the fixed point, $\bar J_{IE}\sum_{\alpha = 1}^n x_\alpha + \bar J_{II}x_{n+1} + b_{I} = x_{n+1} > 0$, so dropping the threshold was appropriate. 
Substituting this back into the equations for the excitatory neurons, we have 

\begin{align*}
\tau_E \dot x_e &= -x_e + \left \lfloor \bar J^\circlearrowleft_{EE} x_e + \sum_{j \to e} \bar J_{EE}^\to x_j  +  \sum_{j \not \to e} \bar J_{EE}^{\not\to} x_j + \frac{\bar J_{EI} \bar J_{IE}}{1 - \bar J_{II}}\sum_{j\in E} x_j + \frac{\bar J_{EI}}{1 - \bar J_{II}} b_I +  b_E\right \rfloor _+\\
\tau_E \dot x_e &= -x_e +\\ 
&\left \lfloor 
\left(\bar J^\circlearrowleft_{EE} + \frac{\bar J_{EI} \bar J_{IE}}{1 - \bar J_{II}}\right)x_e + 
\sum_{j \to e}
\left(\bar J_{EE}^\to + \frac{ \bar J_{EI} \bar J_{IE}}{1 - \bar J_{II}}\right) x_j 
+ \sum_{j\not\to e} \left(\bar J_{EE}^{\not\to} + \frac{\bar J_{EI} \bar J_{IE}}{1 - \bar J_{II}}\right)x_j + \frac{\bar J_{EI}}{1 - \bar J_{II}} b_I +  b_E\right \rfloor _+ 
\end{align*}
Thus, when 
\begin{align*} 
\bar J^\circlearrowleft_{EE} + \frac{\bar J_{EI} \bar J_{IE}}{1 - \bar J_{II}} = 0\\
\bar J_{EE}^\to + \frac{ \bar J_{EI} \bar J_{IE}}{1 - \bar J_{II}} &= -1 + \epsilon\\
\bar J_{EE}^{\not\to} + \frac{ \bar J_{EI} \bar J_{IE}}{1 - \bar J_{II}} &=  -1 -\delta\\
\frac{\bar J_{EI}}{1 - \bar J_{II}} b_I +  b_E &= \theta
\end{align*}
the $n$-neuron E-I model limits to the CTLN as $\tau_I \to 0. $ Notice that this defines a three-dimensional manifold of six-dimensional parameter space.
Figure \ref{fig:mean_field}D shows the cluster average firing rates for excitatory clusters against rates predicted using the CTLN model. Translation from the $x$-model to the $v$-model is done as in Fig. \ref{fig:mean_field}B. 

\subsection{Bounds on total excitatory, inhibitory activity in EI-TLN}
\label{sec:total_bound}
In the previous section, we informally showed that the EI-TLN model approaches the CTLN model in the limit of fast inhibition. Now, we make this argument more formal and show that when $\tau_I \leq \tau_I^*$ for some fixed $\tau_I^* > 0$ (dependent on parameters and network structure), there is a forward-invariant region on which certain properties of the CTLN still hold. 
In particular,  on this region bounds on total excitatory activity from \cite{lienkaemper_combinatorial_2022} (Eq. \ref{eqn:ctln_total_bound}) apply to the EI-TLN model and the inhibitory  population stays above threshold, allowing us to write the dynamics of the excitatory neurons as an  integro-differential equation equation which does not reference the inhibitory firing rate. 

\begin{prop}
   \label{prop:ctln_bound} Let a directed graph $G$, together with parameters $ \bar J_{EE}^\circlearrowleft, \bar J_{EE}^\to, \bar J_{EE}^{\not\to}, \bar J_{EI}, \bar J_{IE}, \bar J_{II}$ define an EI-TLN.   Without loss of generality, assume $\tau_E = 1$.   Define  $x_E^{\min} = \frac{\frac{\bar J_{EI}}{1 - \bar J_{II}} b_I +  b_E}{  \bar J_{EE}^\circlearrowleft - \bar J_{EE}^{\not\to}}$ and  $x_E^{\max} = \frac{\frac{\bar J_{EI}}{1 - \bar J_{II}} b_I +  b_E}{  \bar J_{EE}^\circlearrowleft  - \bar J_{EE}^\to }$.  Then there exists a parameter-dependent threshold value $\tau_I^*$ such that for all $\tau_I \leq \tau_I^*$, there is a forward invariant set $A \subset B \cap T_{I}^+$, where $B$ is the ``box" defined by 
    \begin{align*}
    B =     \left \{\vec x \,\middle \vert\,     x_E^{\min}< \sum_{\alpha \mbox{ exc.}} x_\alpha(t) <    x_E^{\max} \right \} 
        \cap  \left \{\vec x  \,\middle \vert\,      \frac{\bar J_{IE}}{  \left( 1- \bar J_{II}\right)} x_E^{\min} < x_I<    \frac{\bar J_{IE}}{  \left( 1- \bar J_{II}\right)} x_E^{\max} \right \} \cap \{\vec x \mid x_{\alpha} \geq 0, \alpha \in [n+1] \}
    \end{align*}
    and $T_I^+$ is the region where $\dot x_I$ is above threshold, that is,
    \begin{align*}
        T_I^+ = \{\vec x \mid \bar J_{IE} \sum_{\alpha \exc} x_{\alpha} + \bar J_{II}x_I + b_I \geq 0\} .
    \end{align*}
\end{prop}

\begin{figure}[ht!]
    \centering
    \includegraphics{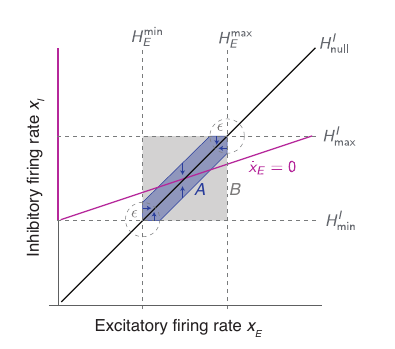}
    \caption{
   Schematic figure for proof of Proposition \ref{prop:ctln_bound}, illustrated for the two-dimensional case of one excitatory population and one inhibitory population. 
    \label{fig:ctln_bound}} 
\end{figure}

\begin{proof}
We construct the region $A \subset B\cap T_I^+$ as follows. First, we define the hyperplanes 

\begin{align*}
    H_E^{\min} = \{\vec x\mid \sum_{\alpha \exc} x_\alpha= x_E^{\min}\} && H_E^{\max} = \{\vec x\mid  \sum_{\alpha \exc} x_\alpha =x_E^{\max}\}&& H_I^{\min} = \{\vec x \mid x_I = x_I^{\min}\}&& H_I^{\max} = \{\vec x \mid x_I = x_I^{\min}\}
\end{align*}
that, together with the coordinate hyperplanes, form the boundary of the box $B$. 
Let  $H_I^{\mathrm{null}}$ denote the portion of the $x_I$ nullcline which is contained in the positive orthant. Notice that $H_I^{\nullh}$ intersects two of the corners of the box $B$: 
$C^{\min} = H_I^{\nullh} \cap H_I^{\max} \cap H_E^{\max} \cap \R^{n+1}_\geq\neq \emptyset$ and 
$C^{\max} = H_I^{\nullh} \cap H_I^{\min} \cap H_E^{\min} \cap\R^{n+1}_\geq\neq \emptyset$. 

Further, notice that on $H_I^{\nullh}$,  for all excitatory neurons $\alpha$, $\dot x_{\alpha}$ is given by the CTLN equations. Thus, we have that $\sum_{\alpha \exc} \dot x_{\alpha} > 0$  on $C^{\min}$ and   $\sum_{\alpha \exc} \dot x_{\alpha} < 0$ on $C^{\max}$ because these inequalities hold in the original CTLN model (\cite{lienkaemper_combinatorial_2022}, Corollary 9.2). 
Because derivatives are continuous in threshold linear networks, this means that for some $\epsilon > 0$,  there are open $\epsilon$-neighborhoods of $C^{\min}$, $C^{\max}$ where $\sum_{\alpha \exc} \dot x_{\alpha} > 0$ and $\sum_{\alpha \exc} \dot x_{\alpha} < 0$ respectively.  Next, notice that by choosing an appropriate $\epsilon,$ we can ensure that both of these neighborhoods are fully contained on the positive side of $T_I^+$, i.e., the region where the inhibitory population is above threshold. 

Thus, we define a hyperplane $H_I^{\nullh, +}$  by translating $H_I^{\nullh}$ upwards (in the positive $x_I$ direction) so that the intersection of $H_I^{\nullh, +}$ and $H_E^{\min}$ remains within this $\epsilon$-neighborhood. 
We define $H_I^{\nullh, -}$ analogously. 

Now, we define $A$ to be the intersection of $B$ with the region bounded by $H_I^{\nullh, +}, H_I^{\nullh, -}$. We show that there exists a threshold value of $\tau_I^*$ such that for all $\tau_I \leq \tau_I^*$, $A$ is a forward-invariant set. To do this, we need to check that for each hyperplane bounding $A$,  the dot product of $(\dot x_1, \ldots, \dot x_n, \dot x_{n+1})$ with the normal vector of the hyperplane (oriented towards the inside of $A$) is positive by for large enough $\tau_I.$ 

There are $6 $ hyperplanes to check: $ H_E^{\min}, H_E^{\max}, H_I^{\min},H_I^{\max},  H_I^{\nullh +},$ and $ H_I^{\nullh -}$. 
First, we note that by the way we have chosen $\epsilon$, $\sum_{\alpha \exc} \dot x_\alpha > 0$  on the portion of $H_E^{\min}$ which forms the boundary of $A$ and  $\sum_{\alpha \exc} \dot x_\alpha < 0$ on the portion of $H_E^{\max}$ which forms the boundary of $A$. Note that this also includes the corners. Thus, the EI-TLN vector field points into $A$ along $H_E^{\min}$ and $H_E^{\max}$. 

Next, we note that, where they form the boundary for $A$, $H_I^{\max}$ is above the $H_I^{\nullh}$ and $H_I^{\min}$ is below $H_I^{\nullh}$, so $\dot x_I$ is positive on $H_I^{\min}$ and negative on $H_I^{\max}$. Further, the normal vector for $H_I^{\max}$ is pointed entirely in the $-x_I$ direction while the normal vector for $H_I^{\min}$ is entirely in the $+x_I$ direction.  Thus, for any value of $\tau_I$, the EI-TLN vector field points in to $A$. 

Finally, we check $H_I^{\nullh+}$ and $H_I^{\nullh-}$. 
Notice that the normal vector for $H_I^{\nullh +}$ ($H_I^{\nullh -}$) has a negative (positive) entry in the $x_I$ coordinate. Further, because $H_I^{\nullh +}$ ($H_I^{\nullh -}$) is above (below) $H_I^{\null}$ , $\dot x_I$ is negative on  $H_I^{\nullh +}$  and positive on $H_I^{\nullh -}$ .  Further, scaling $\tau_I$ while fixing other parameters does not affect the excitatory derivatives $\dot x_1, \ldots, \dot x_n$. Because $A$ is compact,  $|\sum_{\alpha \exc}x_{\alpha}|$  attains a finite maximum on $H_{I}^{\nullh +} \cap A, H_{I}^{\nullh -} \cap A$. Likewise, because $H_I^{\nullh +}, H_I^{\nullh -}$ are disjoint from $H_I^{\nullh}$, $\dot x_I$ attains a non-zero minimum value on $H_{I}^{\nullh +} \cap A, H_{I}^{\nullh -} \cap A$. Thus, by scaling $\tau_I$, it is possible to guarantee that $\dot {\vec x} $ has a positive dot product with $H_{I}^{\nullh +} \cap A$ and $H_{I}^{\nullh -} \cap A$.

\end{proof}
\subsection{Paradoxical response in EI-TLN}
\label{sec:app_paradox}
We prove that in the limit where $\tau_I/\tau_E \to 0$, the EI-TLN exhibits the paradoxical response: excitatory stimulation of the inhibitory population results in a \emph{decrease} in inhibitory firing rates. Although interpretation of this statement is straightforward for systems with a single, inhibition-stabilized fixed point, we need to clarify what this means in a system with a dynamic attractor. What we show is that for any EI-TLN with sufficiently fast inhibition, increasing the input to the inhibitory neurons scales down the inhibitory firing rates on all attractors.

\begin{prop}
Let a directed graph $G$, together with the parameters $ \bar J_{EE}^\circlearrowleft, \bar J_{EE}^\to, \bar J_{EE}^{\not\to}, \bar J_{EI}, \bar J_{IE}, \bar J_{II}$, define an EI-TLN. Let $\tau_I \leq \tau_I^*$, where $\tau_I^*$ is the threshold of inhibitory timescale described in Proposition \ref{prop:ctln_bound}.
Let $\vec x(t)$ be a solution to  the EI-TLN, with external input $b_E, b_I$, which starts in the forward invariant set $A$ defined in Proposition \ref{prop:ctln_bound}. 
A solution $\vec y(t)$ to the EI-TLN with external input $b_E, b_I + \Delta b_I$ can be constructed by scaling the entries of $\vec x(t)$ as:

\begin{align*}
    \vec y(t) = \begin{pmatrix}
        \kappa \vec x_E(t)\\
        \kappa x_I(t)  + \frac{(1 - \kappa)b_I + \Delta b_I}{1 - J_{II}}
    \end{pmatrix}
\end{align*}
where
\begin{align*}
    \kappa = \frac{\frac{\bar J_{EI}}{1 - \bar J_{II}}(b_I + \Delta b_I) + b_E}{\frac{\bar J_{EI}}{1 - \bar J_{II}}b_I  + b_E},
\end{align*}
$\vec x_E(t) = (x_1(t), \ldots, x_{n}(t)) $, and $x_I(t) = x_{n+1}(t)$. When $\Delta b_I \geq 0$, $y_I(t) < x_I(t)$ for all $t$.

\end{prop}

\begin{proof}
    Because $\vec x(0) \in A \subseteq B \cap T_I^+$, $\vec x(T) \in T_I^+$ for all $t\geq 0.$ So along $\vec x(t)$, the inhibitory population always stays above threshold,  meaning that its dynamics are linear. This means we can describe the total inhibitory activity in terms of the external input and the time-average of the total excitatory input in the recent past: 
        \begin{align*}
        x_I(t) = \frac{\bar J_{IE}}{1 - \bar J_{II}} \bar x(t) + \frac{b_I }{1-\bar J_{II}},
    \end{align*}
    where 
    \begin{align*}
    \bar x(t)= \frac{1}{\tau_I}\int_0^t e^{(t'-t)/\tau_I} \sum_{\alpha \exc} x_{\alpha(t')} dt'.
    \end{align*} 
    This makes it possible to write an integro-differential equation for the excitatory rates which does not make reference to the inhibitory rates: 
    \begin{align}
        \dot x_{\alpha} = -x_{\alpha} + \left \lfloor  \sum_{\beta \exc} \bar J_{\alpha \beta} x_{\beta} + \frac{\bar J_{EI}\bar J_{IE}}{1 - \bar J_{II}} \bar x + \frac{\bar J_{EI}b_I }{1-\bar J_{II}} + b_E \right \rfloor _+ 
        \label{eqn:int_diff}
    \end{align}

    Now, for all excitatory neurons $\alpha$, let $y_{\alpha}(t) = \kappa x_{\alpha}(t)$.  Notice that 
\begin{align*}
  -y_{\alpha}(t) + \left \lfloor  \sum_{\beta \exc} \bar J_{\alpha \beta} y_{\beta} + \frac{\bar J_{EI}\bar J_{IE}}{1 - \bar J_{II}} \bar y+ \frac{\bar J_{EI}(b_I + \Delta b_I) }{1-\bar J_{II}} + b_E \right \rfloor _+  = \kappa \dot x_{\alpha} = \dot{y}_{\alpha}
    \end{align*}
     Thus if $x_{\alpha}(t)$ is a solution to Equation \ref{eqn:int_diff}, then $y_{\alpha(t)}$ is a solution to the same threshold linear network with input $b_E, b_I + \Delta b_I$. 
    Along this solution to the threshold linear network, the inhibitory population's firing rate is 

        \begin{align*}
        y_I(t) = \kappa \frac{\bar J_{IE}}{1 - \bar J_{II}} \bar x(t) + \frac{b_I + \Delta b_I}{1-\bar J_{II}} 
    \end{align*}
    Now, we compare $y_I(t)$ to $x_I(t)$.  We have 
    \begin{align*}
        y_I(t) - x_I(t) =& \frac{ 1}{ 1- \bar J_{II}} \left ( (\kappa -1) \bar J_{IE} \bar x(t)  + \Delta b_I\right) \\
        =& \frac{ \Delta b_I }{ 1- \bar J_{II}} \left ( \frac{\frac{\bar J_{EI}}{1 - \bar J_{II}}}{\frac{J_{EI}}{1 - J_{11}} b_I + b_E} \bar J_{IE} \bar x(t)  + 1\right).
    \end{align*}
Now, because $\vec x(0) \in A \subseteq B \cap T_I^+$, $\vec x(t)$ satisfies $\sum_{\alpha \exc} x_{\alpha }(t) >\frac{\frac{J_{EI}}{1 - J_{11}} b_I + b_E}{\bar J_{EE}^\circlearrowleft - \bar J_{EE}^{\not \to}}$. Thus, 

\begin{align*}
    \frac{\bar x(t) }{\frac{J_{EI}}{1 - J_{11}} b_I + b_E} > \frac{1}{\bar J_{EE}^\circlearrowleft - \bar J_{EE}^{\not \to}}.
\end{align*}

Applying the identity  $\bar J_{EE}^\circlearrowleft + \frac{\bar J_{EI} \bar J_{IE}}{1 - \bar J_{II}} = 0, $ we have 
\begin{align*}
    \frac{\frac{\bar J_{EI}}{1 - \bar J_{II}}}{\frac{J_{EI}}{1 - J_{11}} b_I + b_E} \bar J_{IE} \bar x(t) =   \frac{-\bar J_{EE}^\circlearrowleft \bar x(t)}{\frac{J_{EI}}{1 - J_{11}} b_I + b_E}  < \frac{-\bar J_{EE}^\circlearrowleft }{\bar J_{EE}^\circlearrowleft - \bar J_{EE}^{\not\to}} < -1.
\end{align*}
Thus, 
    \begin{align*}
                y_I(t) - x_I(t) < - \frac{ \Delta b_I }{ 1- \bar J_{II}}. 
    \end{align*}
    Increasing the input to the inhibitory population decreases its activity along all trajectories that start in the forward invariant region $A \subseteq B \cap T_I^+$ where the inhibitory neuron stays in a linear regime and the total excitatory activity obeys the bounds for CTLNs. 

\end{proof}

Figure \ref{fig:para} shows the response of the clustered spiking model to stimulation of the inhibitory cluster. In panel A, the network  has a single excitatory population and single inhibitory population. In panel $B$, the network has three excitatatory clusters, whose connectivity is a directed cycle, i.e., the graph has adjacency matrix 
\begin{align*}
    A = \begin{pmatrix}
        0 &0 &1\\
        1 & 0 & 0\\
        0 & 1 & 0
    \end{pmatrix}.
\end{align*}
The simulation parameters are $N_E = N_I = 1000$, $\tau_E = 40  \unit{ms}, \tau_I = 20 \unit{ms}$, $p_{EE}= .2, p_{EI}= .8, p_{IE}= .8,  p_{II}= .8, p_{EE}^\circlearrowleft =  .8$ , $J_{EE}^\circlearrowleft =2$, $J_{EE}^\to =1.25$, $J_{EE}^{\not\to} = 0$, $J_{EI} = -2$, $J_{II} = -3.0$, $J_{IE} = 4.0, b_E = 0.1$, $b_I = 0$. At time $t = 500 \unit{ms}$, an inhibitory stimulus is applied, changing the the value of $b_I$ to $0.15$.

\subsection{How fast does inhibition need to be?}

\subsubsection{Fixed point}
\label{sec:fixed_point_bif}
Now, we explore how fast inhibition must be relative to excitation for stable fixed points supported on target-free cliques to remain stable in the EI-TLN . At the fixed point, all excitatory firing rates have the same value. Further, because the graph is a clique, all pairs of excitatory clusters are connected with an edge of weight $\bar J_{EE}^{\to}$. 
We can therefore reduce to a system with one excitatory unit and one inhibitory unit with effective weights 
\begin{align*}
    \bar J_{EE}^{eff}& = (n-1)\bar J_{EE}^\to + \bar J_{EE}^\circlearrowleft\\
    \bar J_{IE}^{eff}& = n \bar J_{IE}.\\
\end{align*}
The stability of the fixed point depends on the eigenvalues of the matrix 
\begin{align*}
    A = \begin{pmatrix}
        \frac{\bar J^{eff}_{EE}-1}{\tau_E} & \frac{\bar J_{EI}}{\tau_E}\\
        \frac{\bar J^{eff}_{IE}}{\tau_I} & \frac{\bar J_{II}-1}{\tau_I}\\
    \end{pmatrix}
\end{align*}
which are 
\begin{align*}
    \lambda^{\pm} = \frac{-B \pm\sqrt{B^2 - 4AC}}{2A}
\end{align*}
where 
$A =\tau_E\tau_I$, $B = -((\bar J^{eff}_{EE}-1)\tau_I + (\bar J_{II}-1)\tau_E)$, $C = (\bar J^{eff}_{EE}-1)(\bar J_{II}-1) - \bar J_{EI}\bar J^{eff}_{IE}.$  

Note that $C >0 $ for E-I CTLN parameters and any value of $\tau_I, \tau_E$:  
\begin{align*}
    C &= ((n-1) \bar J_{EE}^\to + \bar J_{EE}^\circlearrowleft -1)(\bar J_{II}-1)  - n\bar J_{EI}\bar J_{II}\\
 & = (\bar J_{II}-1)  \left((n-1) \bar J_{EE}^\to + \bar J_{EE}^\circlearrowleft -1 - n\frac{\bar J_{EI}\bar J_{II}}{\bar J_{II}-1}\right)\\
 & =  (\bar J_{II}-1)  \left((n-1) \bar J_{EE}^\to - (n-1)\bar J_{EE}^\circlearrowleft- 1 \right) \\
 &> 0.
     \end{align*}

This implies that there is a Hopf-like bifurcation when $B = 0,$ and that the fixed point is  stable when 
\begin{align*}
    -\frac{\bar J^{eff}_{EE}-1}{\bar J_{II} -1} < \frac{\tau_E}{\tau_I}.
\end{align*}
Because the system is non-smooth, the bifurcation does not fit the normal form for a standard bifurcation. We numerically continued the limit cycle which arises from this fixed point using the Julia package BifurcationKit \cite{veltz_bifurcationkitjl_2020}. At the parameters used in Fig.~\ref{fig:fp_bif} ($k = 1, \, \bar J_{EE}^\circlearrowleft = 1.5, \, \bar J_{EE}^\to = .75, \bar J_{EE}^{\not\to} = 0,\, \bar J_{EI} = -1.5, \, \bar J_{IE} = 2, \, \bar J_{II} = -1$) the bifurcation is a non-smooth analogue of a subcritical Hopf bifurcation since there is a small region of bistability ($3.8 \leq \tau_I/\tau_E \leq 4 $ of bistability between the limit cycle and the fixed point. 

\subsubsection{Limit cycle }
\label{sec:limit_cycle_bif}
Now, we consider how fast inhibition must be relative to excitation for a limit cycle in a CTLN to persist in the EI-TLN. We investigate this numerically using the Julia package BifucationKit for continuation of periodic orbits  \cite{veltz_bifurcationkitjl_2020}. For this analysis, we used the parameters  $J_{EE}^\circlearrowleft = 1.5, J_{EE}^\to =  0.75, J_{EE}^{\not\to} = 0.0, J_{EI} = -2.25, J_{IE} = 2.0, J_{II} = -2.0, \tau_E = 1.0$ We considered the directed graph with adjacency matrix 
\begin{align*}
A = \begin{pmatrix}
        0 & 0 & 1\\
    1 & 0 & 0\\
    0 & 1 & 0\\
\end{pmatrix}.
\end{align*}

To identify the CTLN-like cycle (Fig.~\ref{fig:3cycle_bif}, panels B and D (left), gold line in panel C), we first simulated the EI-TLN with $\tau_I = 0.01$ with the initial condition $x_0^{\mathrm{CTLN}} = [0.5, 0, 0.5, 0.67].$    
We used this trajectory as an initial guess for periodic orbit localization in BifurcationKit. Our initial guess for the period of the orbit was $T = 11$. To localize and continue this limit cycle, and all others in this section,  we discretized time into $m = 100$ time steps and used a finite differences method based on the trapezoid rule, as implemented in BifucationKit with the function \texttt{generate\_ci\_problem} in BifurcationKit. 
We continued this periodic orbit from $\tau_I = 0.001$ to $\tau_I = 8.0$ to produce the bifurcation diagram and other plots in Fig.~\ref{fig:3cycle_bif}.   Likewise, to localize and continue the EI cycle,  we made an initial guess by simulating the EI-TLN with $ \tau_I = 2.1$ and initial condition $x_0^{\mathrm{CTLN}} = [0.5, 0.5, 0.5, 0.7]$, and again continued a periodic orbit based on this initial guess, with initial guess of  $T = 5$ for the period. 

Finally, to localize and continue the mixed cycle, we first identified an analogous cycle in a smoothed version of the EI-TLN, with the non-smooth threshold linear function $[x]_+ = \max\{x, 0\}$ replaced by the smooth function $\texttt{squareplus}(x) = \frac{x + \sqrt{x^2 + a}}{2}$ for $a = 0.0001$. To identify the unstable limit cycle in this system, we made an initial guess by simulating the dynamics for the initial conditions $x_0^{\mathrm{mixed}} = 0.225 x_0^{\mathrm{CTLN}} + (1 - 0.225) x_0^{\mathrm{EI}}$
and $\tau_I = 2.5$.  In the smoothed system, continuing this limit cycle results in a loop in the bifurcation diagram. To identify an analogous limit cycle in the non-smooth case, we choose periodic orbits from the top and bottom branches of this loop. We then localize and continue each of these periodic orbits in the non-smooth system. In the original non-smooth system, the top and bottom branches of this periodic orbit meet at a cusp in the bifurcation diagram, so it is not possible to continue around the loop in that system. 

\end{document}